\documentclass[12pt]{book}

\usepackage[a4paper,margin=1in]{geometry}

\usepackage[T1]{fontenc}
\usepackage[utf8]{inputenc}
\usepackage{lmodern}

\usepackage{amsmath, amssymb} 
\usepackage{titlesec}
\usepackage{enumitem}

\usepackage{tabularx}   
\usepackage{array}      
\usepackage{graphicx}   
\usepackage{caption}    

\usepackage{hyperref}

\usepackage{amsthm}
\newtheorem{lemma}{Lemma}

\usepackage[most]{tcolorbox}
\usepackage{xcolor}

\definecolor{humanblue}{RGB}{0,90,160}
\definecolor{aigreen}{RGB}{0,120,0}

\usepackage{hyperref}

\hypersetup{
  colorlinks=true,
  linkcolor=blue
}

\newtcolorbox{humanprompt}{
  colback=humanblue!5,
  colframe=humanblue,
  title=Human Prompt,
  fonttitle=\bfseries,
  sharp corners,
  boxrule=1pt,
  breakable
}

\newtcolorbox{aireply}{
  colback=aigreen!5,
  colframe=aigreen,
  title=AI Reply,
  fonttitle=\bfseries,
  sharp corners,
  boxrule=1pt,
  breakable
}

\titleformat{\section}{\large\bfseries}{\thesection}{1em}{}
\titleformat{\subsection}{\normalsize\bfseries}{\thesubsection}{1em}{}

\title{Enhancing Formal Software Specification with Artificial
Intelligence}
\author{
Antonio Abu Nassar,
Eitan Farchi  \\
\textit{IBM Research}\\
antonio.abu.nassar@ibm.com,
farchi@il.ibm.com
}
\date{January 2026}

\begin{document}
\maketitle



\begin{center}
\textbf{Abstract}
\end{center}

Formal software specification is known to enable early error detection
and explicit invariants, yet it has seen limited industrial adoption
due to its high notation overhead and the expertise required to use
traditional formal languages. This paper presents a case study showing
that recent advances in artificial intelligence make it possible to
retain many of the benefits of formal specification while
substantially reducing these costs.
The necessity of a clear distinction between what is controlled
by the system analyst and can highly benefits from the rigor of formal specification and what need not be controlled is demonstrated.
We use natural language augmented
with lightweight mathematical notation and written in \LaTeX\ as an
intermediate specification language, which is reviewed and refined by
AI prior to code generation. Applied to a nontrivial simulation of
organizational knowledge growth, this approach enables early
validation, explicit invariants, and correctness by design, while
significantly reducing development effort and producing a correct
implementation on the first attempt.


\tableofcontents
\newpage

\section{Introduction}

The benefits of formal software specification—such as early detection
of design flaws, explicit assumptions, and built-in invariants—are
widely acknowledged in both academia and industry. Nevertheless,
formal methods have seen limited industrial adoption\cite{Saiedian1996}. The primary
barriers are the significant overhead involved in writing and
maintaining formal specifications, as well as the specialized
expertise required to use traditional specification languages
(e.g., Z, VDM). As a result, many systems continue to rely on informal
natural-language requirements, deferring correctness concerns to
later development and testing stages.

Recent advances in artificial intelligence raise a natural and timely
question: \emph{can we obtain many of the benefits of formal
specification while avoiding or substantially reducing its
traditional costs}? This paper presents a case study demonstrating
that the answer is affirmative. We show that natural language, when
augmented with lightweight mathematical notation and written in
\LaTeX, can serve as an effective intermediate specification
language. When combined with AI-based review, reflection, and code
generation, this approach enables early validation, explicit
invariants, and a strong form of correctness-by-design—without the
burden of heavyweight formal methods.

Our case study focuses on a simulation of knowledge growth within an
organization. This domain was deliberately chosen because it is
sufficiently rich to require nontrivial structure, invariants, and
parameterization, while still being accessible and interpretable. The
simulation therefore serves as a realistic and demanding testbed for
evaluating AI-assisted specification and development.

A central goal of this work is to \emph{shift validation left} as much
as possible in the development lifecycle. We leverage AI throughout
the specification, design, and validation process, using natural
language as the primary interface between the human and the system.
Core requirements and design decisions are refined through AI-driven
reflection and other agentic design patterns, with the aim of
producing a system that is inherently correct by construction rather
than corrected post hoc through debugging.

We examine two distinct methodologies for developing the simulation
from a natural-language specification. In the first approach, the
simulation is corrected iteratively based solely on its execution
results. In the second approach, AI is used to review and refine the
specification \emph{before} any code is generated or executed.
Figure~\ref{fig:simulation-methodology} illustrates these two
workflows. Empirically, the second approach required approximately
$\frac{1}{6}$ of the time needed by the first and produced a correct
simulation on the first code-generation attempt, highlighting the
practical value of early, AI-assisted specification review.

\begin{figure}[htbp]
    \centering
    \includegraphics[width=0.8\textwidth]
        {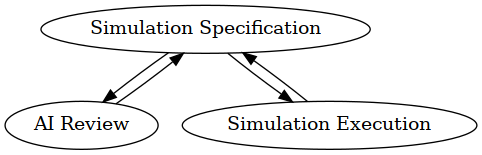}
    \caption{Simulation methodology flow showing iterative connections
    between specification, AI review, and execution.}
    \label{fig:simulation-methodology}
\end{figure}

The specification style used in this work can be viewed as an
intermediate representation between informal natural language and
fully formal specification languages. Natural language is enriched
with mathematical notation, enabling precise assertions that can be
reasoned about, proved at the specification level, and later
translated into runtime invariants.

An important observation from the study is that, while high-level
invariants were envisioned by the human designer, the AI often
proposed lower-level invariants that were both correct and useful.
This division of labor proved effective: the human focused on
conceptual intent, while the AI contributed detailed consistency
checks. However, we also encountered attention limitations. As the
specification grew beyond a few pages, the AI occasionally omitted
parts of it during generation. This reinforces the importance of
abstraction, modularization, and decomposition when working with
longer specifications.

Beyond faithfully implementing stated requirements, the AI frequently
suggested natural and valuable extensions that had not been explicitly
requested. In a business setting, such suggestions might include
internationalization or configurability. In our case study, one
notable example was the proposal to introduce Monte Carlo simulation
to ensure statistical stability of the results. These contributions
were welcomed, provided they could be validated at the specification
level.

Throughout the development of the simulation, the generated code
itself was not manually reviewed, as the combination of
specification-level reasoning and invariant checking proved
sufficient. The final system supports multiple experimental
configurations, automatically compares alternative strategies, and
produces succinct conclusions about their relative effectiveness.
Interestingly, prompting the AI to explicitly summarize experimental
conclusions required some repetition, suggesting an area for further
refinement of agentic workflows.

An additional experiment, described in Section~\ref{sec:ambiguous_business_specification},
used summarization of the refined, carefully written specification of the business logic,
to obtain a version of the specification where the line between what the system analyst
wants to control and what the analyst does not want to control became blurred.
The experiment emphasizes the primacy of precision when writing
specifications for business logic, and the importance of a clear distinction between
areas of strict control and areas where control is not desired.

All experiments reported in this paper were conducted using AI.
The generated code is written in Python and is available via linked
Google Colab notebooks. The documented sequence of prompts and
responses can serve as a benchmark for comparing the capabilities of
different AI agents in specification-driven software development.

For readability, the reader is encouraged to first follow the
high-level development flow, skimming the prompts and AI responses,
before diving into the details. The table of contents provides a clear
overview of the interaction structure, and visual boxes are used
throughout the text to distinguish between human prompts (blue) and
AI replies (red).

\section{Experiment – Codeless Development}

To gain experience in creating software solely based on its behavior and refining its natural language specification, I defined a game that simulates the impact of a \textit{Tertius Gaudens} (the third who links)—a person in an organization who connects other individuals or groups—on the growth of knowledge within the organization. The basic prompt used to initiate this experiment is available in \ref{basicPrompt}.

\vspace{1em}
I dedicated approximately four hours to this exercise. The interaction resulted in the generation of hundreds of lines of code, which can be viewed here:  
\href{https://colab.research.google.com/drive/1qFuYmyzGHcDNqN-rTAgYEA_L8LuBgEbL?usp=sharing}{\texttt{Notebook}}.

\vspace{0.5em}
Additionally, a 7-page report was produced:  
\href{https://www.overleaf.com/read/rppcmvdhgzvw#89c2ca}{\texttt{Report}}.

\vspace{1em}
Interestingly, I did not inspect the code directly during the development process. Instead, I relied on iterative prompting and reviewing simulation results to guide the development. The report was validated through careful reading and minor adjustments.  I estimate that creating the simulation and conducting the research from scratch would have taken me a week.  




\subsection{The base prompt}
\label{basicPrompt}


\begin{humanprompt}

\textbf{Prompt:} \\
Given the following definition of the simulation, will you be able to write tests for it?

\medskip

We are given a graph \( G \). The vertices are a set of inventors \( v_1, v_2, \ldots, v_n \). An edge \( (v_i, v_j) \) means that \( v_i \) and \( v_j \) communicate. Each inventor is associated with a set of invention items, initialized to the empty set.

\medskip

The game is played in stages. In each stage, another player called \texttt{tg} can decide to add \( k \) edges to the graph, where \( k < n^2 \). For example, consider \( k = 1 \).

\medskip

Next, each clique produces a new knowledge item with probability:
\[
\frac{\text{Number of knowledge items owned by clique members}}{\text{Total number of knowledge items}}
\]
If a new knowledge item is created, it is added to the knowledge set of each clique member.

\medskip

If the clique that created a knowledge item includes an edge added by \texttt{tg}, then \texttt{tg} gains a reward of 1.

\medskip

The game is repeated starting from the initial graph.

\medskip

\textbf{Task:} \\
Create a simulation of the game in Python for five players:
\begin{itemize}
    \item \( v_1 \) connected to \( v_2 \) and \( v_3 \)
    \item \( v_2 \) connected to \( v_3 \)
    \item \( v_3 \) connected to \( v_4 \) and \( v_5 \)
    \item \( v_4 \) connected to \( v_5 \)
\end{itemize}

Include a \texttt{tg} player with \( k = 1 \). Play the game for 10 stages. Represent plots of:
\begin{itemize}
    \item The gain of \texttt{tg}
    \item The knowledge growth of each player
\end{itemize}

Let \texttt{tg} choose an edge to add:
\begin{itemize}
    \item Randomly in one simulation
    \item Always connecting \( v_1 \) and \( v_4 \) in another simulation
\end{itemize}

\end{humanprompt}

\section{Simulation Definition Review}

By iterating with the AI, I addressed inconsistencies it identified during the review process. The iterations below resulted in a specification that is consistent and a simulation that ran correctly on the first attempt. Notably, this process took only 20 minutes. In the implementation shown below, I focused solely on the first stage of the simulation—previously implemented—based on insights from the simulation execution. As a result, the implementation time was reduced to just $\frac{1}{6}$ of the original effort.

\subsection{User Question}

\begin{humanprompt}
\textbf{Prompt:} \\
Please review the simulation definition above (\ref{basicPrompt}) — is there any contradiction in the definition of the simulation?
\end{humanprompt}

\subsection{Review and Analysis}

\begin{aireply}

\textbf{AI reply:} \\

\textbf{Simulation Summary:}
\begin{itemize}
    \item \textbf{Graph Structure:} A graph \( G \) with vertices \( v_1, v_2, \ldots, v_n \) representing inventors. An edge \( (v_i, v_j) \) indicates communication between inventors.
    \item \textbf{Knowledge Sets:} Each inventor starts with an empty set of invention items.
    \item \textbf{Stages:} In each stage:
    \begin{itemize}
        \item A player called \texttt{tg} adds \( k \) edges to the graph, where \( k < n^2 \).
        \item Each clique in the graph may produce a new knowledge item with probability:
        \[
        \frac{\text{Number of knowledge items owned by clique members}}{\text{Total number of knowledge items}}
        \]
        \item If a new item is created, it is added to the knowledge set of each clique member.
        \item If the clique includes an edge added by \texttt{tg}, \texttt{tg} gains a reward of 1.
    \end{itemize}
    \item \textbf{Repetition:} The game is repeated starting from the initial graph.
\end{itemize}

\textbf{Consistency and validity:}
\begin{itemize}
    \item The rules for knowledge generation and reward assignment are logically consistent.
    \item The probability formula is well-defined and bounded between 0 and 1.
    \item The reward mechanism for \texttt{tg} is clearly tied to the influence of added edges.
\end{itemize}

\textbf{Potential Ambiguities:}
\begin{enumerate}
    \item \textbf{Graph Reset Clause:} The phrase ``game is repeated starting from the initial graph'' suggests the graph resets each stage. This contradicts the idea of cumulative edge additions and knowledge growth. It is recommended to clarify that the graph evolves cumulatively.
    \item \textbf{Probability Edge Case:} When the total number of knowledge items is zero, the formula becomes undefined. A default probability of 1.0 is a reasonable resolution.
    \item \textbf{Clique Definition:} It is assumed that cliques are maximal. Clarification is needed on whether overlapping or nested cliques are considered.
    \item \textbf{Reward Multiplicity:} It is not specified whether multiple rewards can be earned from a single clique if multiple \texttt{tg}-added edges are present. A reasonable interpretation is one reward per clique per stage.
\end{enumerate}

\textbf{Conclusion:}
There are no logical contradictions in the simulation definition. However, the graph reset clause introduces ambiguity and should be revised to reflect cumulative evolution of the graph and knowledge sets.

\end{aireply}

\subsection{User clarifications based on the above}

\begin{humanprompt}

\textbf{Prompt:} \\
Here are the following clarification.
\begin{enumerate}
    \item The TG add edges temporarily at each stage of the simulation according to its strategy.   Then, they are removed returning the graph $G$ to its original state. 
    \item At each stage, the simulation chooses a single clique randomly.  Then a new knowledge item is created for the chosen clique in the prescribed probability.
    \item To avoid the edge cases we will have a knowledge item created in the updated probability of 1 if the total number of knowledge items is 0.

\end{enumerate}

Please update the original specification to incorporate the above clarifications. 

\subsection{The updated prompt}

The following updated prompt was produced by the AI. 

\medskip

We are given a graph \( G \). The vertices are a set of inventors \( v_1, v_2, \ldots, v_n \). An edge \( (v_i, v_j) \) means that \( v_i \) and \( v_j \) communicate. Each inventor is associated with a set of invention items, initialized to the empty set.

\medskip

The game is played in stages. In each stage:
\begin{itemize}
    \item A player called \texttt{tg} can decide to add \( k \) edges to the graph, where \( k < n^2 \). For example, consider \( k = 1 \).
    \item The edges added by \texttt{tg} are \textbf{temporary} and are removed at the end of the stage, returning the graph \( G \) to its original state.
    \item A single clique is selected \textbf{randomly} from the graph (with the temporary edges included).
    \item A new knowledge item is created for the selected clique with probability:
    \[
    \frac{\text{Number of knowledge items owned by clique members}}{\text{Total number of knowledge items}}
    \]
    \item If the total number of knowledge items is zero, the probability is defined to be 1 to avoid edge cases.
    \item If a new knowledge item is created, it is added to the knowledge set of each clique member.
    \item If the selected clique includes an edge added by \texttt{tg}, then \texttt{tg} gains a reward of 1.
\end{itemize}

\medskip

\textbf{Task:} \\
Create a simulation of the game in Python for five players:
\begin{itemize}
    \item \( v_1 \) connected to \( v_2 \) and \( v_3 \)
    \item \( v_2 \) connected to \( v_3 \)
    \item \( v_3 \) connected to \( v_4 \) and \( v_5 \)
    \item \( v_4 \) connected to \( v_5 \)
\end{itemize}

Include a \texttt{tg} player with \( k = 1 \). Play the game for 10 stages. Represent plots of:
\begin{itemize}
    \item The gain of \texttt{tg}
    \item The knowledge growth of each player
\end{itemize}

Let \texttt{tg} choose an edge to add:
\begin{itemize}
    \item Randomly in one simulation
    \item Always connecting \( v_1 \) and \( v_4 \) in another simulation
\end{itemize}

\end{humanprompt}

\subsection{User Question}

\begin{humanprompt}

\textbf{Prompt:} \\
Please review the updated simulation definition above — is there any contradiction in the definition of the simulation?
\end{humanprompt}

\subsection{Consistency and Validity of updated simulation definition}

\begin{aireply}

\begin{itemize}
    \item The updated definition resolves the earlier ambiguity regarding whether the graph evolves or resets. It is now clear that the graph structure is static, and \texttt{tg}'s influence is temporary per stage.
    \item The rule of selecting a single clique per stage simplifies the simulation and avoids overlapping reward conditions.
    \item The edge case for zero knowledge items is explicitly handled, ensuring the probability is always well-defined.
    \item The reward mechanism is clearly scoped to the selected clique and the temporary edges, avoiding ambiguity.
\end{itemize}

\subsubsection*{Conclusion}

The updated simulation definition is logically consistent and unambiguous. All previously noted contradictions or unclear elements have been addressed through the clarifications. The simulation is now well-defined and suitable for implementation and testing.

\end{aireply}

\subsection{Creation of the implementation}

\begin{humanprompt}

\textbf{Prompt:} \\
Great.  Now please create a Python program that implements the simulation
\end{humanprompt}

\subsection{Generated code}

Here is the generated code \href{https://colab.research.google.com/drive/1qFuYmyzGHcDNqN-rTAgYEA_L8LuBgEbL#scrollTo=GtcV9_U6GhXg&line=2&uniqifier=1}{link}.  According to the simulation results code was correct at first iteration.  This is in contrast to the first experiment in which the base prompt was was use to product the code and then the simulation results were used to determine how the simulations needs to be corrected. 

\section{Adding Shapley Value in Cooperative Recommendation Games}


Next, I attempted to incorporate an additional measure of knowledge,
the Shapley value, into the simulation of knowledge growth within the
organization. I began with a basic prompt to guide the AI in this task,
while also providing the prompts used in the previous stages to maintain
context and continuity in the development of the simulation.

\subsection{Basic Shapely value prompt}

\begin{humanprompt}

\textbf{Prompt:} \\
In the context of cooperative games, the "recommendation" game can be formulated as a cooperative game $(N, v)$, where each player $i \in N$ is able to recommend a set of items $B_i$. The entire set of items being considered is denoted by $B = \bigcup_{i \in N} B_i$. The value of a coalition $S \subseteq N$ is defined as:

\[
v(S) = \left| \bigcup_{i \in S} B_i \right|
\]

The Shapley value of a player $i$ in this game is given by:

\[
\phi_i = \sum_{j \in B_i} \frac{1}{|\{k \mid j \in B_k\}|}
\]

This formulation is intuitive: if item $j$ is recommended by player $i$ (i.e., $j \in B_i$), the contribution of $i$ to the recommendation of $j$ is weaker the more other agents also recommend $j$, that is, the larger the set $\{k \mid j \in B_k\}$.

In our context, each player corresponds to a vertex $v_i$ in a graph, and each vertex has an associated set of knowledge items $B_i$. This naturally forms a cooperative game as described above. We generalize the current program to compute the Shapley value of each vertex/player and graphically represent its change during the simulation.

\end{humanprompt}

\subsection{Consistency}

\begin{humanprompt}

\textbf{Prompt:} \\
Please review the following simulation specification and highlight ambiguities and inconsistencies concerning the Shapely value 

\end{humanprompt}

\subsection{Ambiguities and Inconsistencies Concerning the Shapley Value}

\begin{aireply}

This section outlines key ambiguities and inconsistencies in the simulation specification with respect to the application of the Shapley value.

\subsubsection{Mismatch Between Simulation Dynamics and Shapley Value Definition}
The simulation involves multiple stages, temporary edge additions, random clique selection, and probabilistic knowledge creation. In contrast, the Shapley value formulation assumes a static cooperative game with a well-defined value function. It is unclear when the Shapley value is computed (e.g., after each stage or only at the end), and whether it reflects dynamic contributions or static outcomes.

\subsubsection{Unclear Definition of Coalition Value \( v(S) \)}
The cooperative game defines the value of a coalition \( S \subseteq N \) as \( v(S) = \left| \bigcup_{i \in S} B_i \right| \), yet the simulation does not explicitly form or evaluate coalitions. Clique selection is random and temporary, not strategic or player-driven, and there is no mechanism for evaluating the value of arbitrary subsets \( S \subseteq N \) during the simulation.

\subsubsection{Role of \texttt{tg} in the Cooperative Game}
It is ambiguous whether \texttt{tg} is a player in the cooperative game (i.e., included in \( N \)), an external facilitator whose reward is tracked separately, or a contributor to knowledge items. If \texttt{tg} is part of the game, its knowledge set and contribution to \( v(S) \) must be defined.

\subsubsection{Knowledge Item Ownership and Duplication}
The Shapley value formula depends on how many players own each item. It is unclear whether knowledge items are unique per player or shared. If a knowledge item is added to all clique members, duplication affects the denominator in the Shapley formula, potentially inflating or deflating contributions.

\subsubsection{Clique Formation and Coalition Interpretation}
Cliques are selected randomly and temporarily, but Shapley value requires evaluating all possible coalitions. The simulation does not explore all subsets \( S \subseteq N \), nor does it track marginal contributions across permutations.

\subsubsection{Temporal vs. Static Analysis}
The simulation is dynamic (10 stages), while the Shapley value is typically computed for a single snapshot of the game. If the Shapley value is to be tracked over time, it must be recomputed after each stage, the evolution of \( B_i \) must be recorded, and the interpretation of changing Shapley values must be clarified (e.g., does it reflect influence, contribution, or centrality?).

\subsection{Recommendations to resolve ambiguities}
\begin{enumerate}
    \item Clarify the time of Shapley value computation (e.g., per stage or at the end).
    \item Define the role of \texttt{tg} in the cooperative game.
    \item Track knowledge item ownership precisely, ensuring unique identifiers and ownership sets.
    \item Align simulation with coalition evaluation by considering multiple subsets \( S \subseteq N \).
    \item Add a mechanism to track marginal contributions for each player.
    \item Visualize Shapley values per player after each stage to show dynamic contribution.
\end{enumerate}

\end{aireply}

\subsection{Clarifying ambiguity}

\begin{humanprompt}

\textbf{Prompt:} \\
Please consider the following clarifications.
\begin{itemize}
\item 
Calculate the Shapely value after each stage of the simulation.  For each player or vertices, $v_i$, use its knowledge items after the stage as the set $B_{v_i}$ in the definition of the general cooperative game.  
\item 
$tg$ is not part of the cooperative game. The cooperative game's players are only the vertices of the graph $G$.
\item 
When a clique creates a know knowledge item it is added to each of the members of the clique, thus, they share the new knowledge item.  Let $S$ be a clique, assume that $B_{v_i}$ is the list of current knowledge items of $v_i \in S$.  Next assume that a new knowledge item, $j$, was created then we update $B_{v_i} = B_{v_i} \cup j$ for all $v_i \in S$.  thus, after the update each member of $S$ share the knew knowledge item $j$. 
\item 
After a stage in the simulation is completed then we consider the set of vertices and their associated knowledge items as a recommendation cooperative game and use the formula

\[
\phi_i = \sum_{j \in B_i} \frac{1}{|\{k \mid j \in B_k\}|}
\]

to calculate the Shapely value of each vertex.

\item 

\end{itemize}

Please create an updated prompt of the entire simulation that incorporates the above clarifications. 

\end{humanprompt}

\subsection{Corrected Simulation Definition with Shapley Value Integration}

\begin{aireply}

We are given a graph \( G \) whose vertices represent a set of inventors \( v_1, v_2, \ldots, v_n \). An edge \( (v_i, v_j) \) indicates that \( v_i \) and \( v_j \) can communicate. Each inventor \( v_i \) is associated with a set of knowledge items \( B_i \), initialized to the empty set.

\medskip

The game is played in discrete stages. In each stage:

\begin{itemize}
    \item A facilitator player, denoted \texttt{tg}, may add up to \( k \) temporary edges to the graph, where \( k < n^2 \). These edges exist only for the duration of the stage and are removed afterward.
    \item A single clique is selected randomly from the graph, including the temporary edges.
    \item A new knowledge item is created for the selected clique with probability:
    \[
    P = \frac{\left| \bigcup_{i \in C} B_i \right|}{\left| \bigcup_{i \in N} B_i \right|}
    \]
    where \( C \) is the set of clique members and \( N \) is the set of all players. If the denominator is zero, the probability is defined to be 1.
    \item If a new knowledge item is created, it is added to the knowledge set \( B_i \) of each clique member \( i \in C \).
    \item If the selected clique contains any edge added by \texttt{tg}, then \texttt{tg} receives a reward of 1.
\end{itemize}

\medskip

To evaluate contributions using the Shapley value, we define the cooperative game \( (N, v) \), where:
\[
v(S) = \left| \bigcup_{i \in S} B_i \right|
\]
for any coalition \( S \subseteq N \). The Shapley value of player \( i \in N \) is computed as:
\[
\phi_i = \sum_{j \in B_i} \frac{1}{|\{k \in N \mid j \in B_k\}|}
\]
This reflects the fractional contribution of each player to the collective knowledge pool, accounting for shared ownership.

\medskip

\textbf{Simulation Task:}
\begin{itemize}
    \item Initialize a graph with five players:
    \begin{itemize}
        \item \( v_1 \) connected to \( v_2 \) and \( v_3 \)
        \item \( v_2 \) connected to \( v_3 \)
        \item \( v_3 \) connected to \( v_4 \) and \( v_5 \)
        \item \( v_4 \) connected to \( v_5 \)
    \end{itemize}
    \item Include a \texttt{tg} player with \( k = 1 \)
    \item Run the simulation for 10 stages
    \item In one simulation, let \texttt{tg} choose an edge randomly
    \item In another simulation, let \texttt{tg} always add the edge \( (v_1, v_4) \)
    \item After each stage, compute and plot:
    \begin{itemize}
        \item The cumulative reward of \texttt{tg}
        \item The knowledge growth of each player
        \item The Shapley value \( \phi_i \) of each player
    \end{itemize}
\end{itemize}

\subsection{Updated simulation}

The simulation generated using the updated prompt can be seen here \href{https://colab.research.google.com/drive/1qFuYmyzGHcDNqN-rTAgYEA_L8LuBgEbL#scrollTo=eTli_ySY_Vdv&line=1&uniqifier=1}{link}.  Did not have to correct the specification of the simulation.

\section{Putting it together - Shapely value and the two simulation tasks}

Here I attempted to put together the Shapely value specification obtained above and the two simulation tasks.  

\subsection{Corrected simulation specification with Shapley Value and the two simulation tasks}

Next I attempt a complete specification that includes the Shapely value and the two simulations. 

We are given a graph \( G \) whose vertices represent a set of inventors \( v_1, v_2, \ldots, v_n \). An edge \( (v_i, v_j) \) indicates that \( v_i \) and \( v_j \) can communicate. Each inventor \( v_i \) is associated with a set of knowledge items \( B_i \), initialized to the empty set.

\medskip

The game is played in discrete stages. In each stage:

\begin{itemize}
    \item A facilitator player, denoted \texttt{tg}, may add up to \( k \) temporary edges to the graph, where \( k < n^2 \). These edges exist only for the duration of the stage and are removed afterward.
    \item A single clique is selected randomly from the graph, including the temporary edges.
    \item A new knowledge item is created for the selected clique with probability:
    \[
    P = \frac{\left| \bigcup_{i \in C} B_i \right|}{\left| \bigcup_{i \in N} B_i \right|}
    \]
    where \( C \) is the set of clique members and \( N \) is the set of all players. If the denominator is zero, the probability is defined to be 1.
    \item If a new knowledge item is created, it is added to the knowledge set \( B_i \) of each clique member \( i \in C \).
    \item If the selected clique contains any edge added by \texttt{tg}, then \texttt{tg} receives a reward of 1.
\end{itemize}

\medskip

To evaluate contributions using the Shapley value, we define the cooperative game \( (N, v) \), where:
\[
v(S) = \left| \bigcup_{i \in S} B_i \right|
\]
for any coalition \( S \subseteq N \). The Shapley value of player \( i \in N \) is computed as:
\[
\phi_i = \sum_{j \in B_i} \frac{1}{|\{k \in N \mid j \in B_k\}|}
\]
This reflects the fractional contribution of each player to the collective knowledge pool, accounting for shared ownership.

\medskip

\textbf{Simulation Task one:}
\begin{itemize}
    \item Initialize a graph with five players:
    \begin{itemize}
        \item \( v_1 \) connected to \( v_2 \) and \( v_3 \)
        \item \( v_2 \) connected to \( v_3 \)
        \item \( v_3 \) connected to \( v_4 \) and \( v_5 \)
        \item \( v_4 \) connected to \( v_5 \)
    \end{itemize}
    \item Include a \texttt{tg} player with \( k = 1 \)
    \item Run the simulation for 10 stages
    \item In one simulation, let \texttt{tg} choose an edge randomly
    \item In another simulation, let \texttt{tg} always add the edge \( (v_1, v_4) \)
    \item Compute and plot:
    \begin{itemize}
        \item The cumulative reward of \texttt{tg} after each stage
        \item The knowledge growth of each player after each stage
        \item The Shapley value \( \phi_i \) of each player after each stage
    \end{itemize}
\end{itemize}

\textbf{Simulation Task two:}
\begin{itemize}
    \item Initialize a graph with ten players:
    \begin{itemize}
        \item \( v_1 \), \( v_2 \) to \( v_5 \) connected to each other.
        \item \( v_5 \), \( v_6 \) to  \( v_10 \) connected to each other. 
    \end{itemize}
    \item Include a \texttt{tg} player with \( k = 4 \)
    \item Run the simulation for 10 stages
    \item In one simulation, let \texttt{tg} choose to connect 5 pairs of vertices that are not connected.  For example, \texttt{tg} may choose to connect \( v_1 \) with \( v_{10} \).
    \item In another simulation, let \texttt{tg} always connect \( v_4 \) with \( v_6 \) to. \( v_{10} \).
    \item Compute and plot:
    \begin{itemize}
        \item The cumulative reward of \texttt{tg} after each stage.
        \item The knowledge growth of each player after each stage.
        \item The Shapley value \( \phi_i \) of each player after each stage. 
    \end{itemize}
\end{itemize}


\end{aireply}




\subsection{Generated Code}

The AI-generated code resulted in 362 lines, including documentation.
It can be accessed here: \href{https://colab.research.google.com/drive/1qFuYmyzGHcDNqN-rTAgYEA_L8LuBgEbL#scrollTo=Y68iPB6wvlf0&line=362&uniqifier=1}{Colab link}.
The implementation now correctly handles both scenarios—large and
small two-connected cliques—and incorporates the Shapley value
calculation, producing a correct simulation on the first generation
from the specification described above.

\section{Leaving the organization}

Next I attempt to add to the simulation the possibility that a person leaves the organization.

\subsection{Leaving the organization specification}

\begin{humanprompt}

\textbf{Prompt:} \\
We are given a graph \( G \) whose vertices represent a set of inventors \( v_1, v_2, \ldots, v_n \). An edge \( (v_i, v_j) \) indicates that \( v_i \) and \( v_j \) can communicate. Each inventor \( v_i \) is associated with a set of knowledge items \( B_i \), initialized to the empty set.

\medskip

The game is played in discrete stages. In each stage:

\begin{itemize}
    \item A facilitator player, denoted \texttt{tg}, may add up to \( k \) temporary edges to the graph, where \( k < n^2 \). These edges exist only for the duration of the stage and are removed afterward.
    \item A single clique is selected randomly from the graph, including the temporary edges.
    \item A new knowledge item is created for the selected clique with probability:
    \[
    P = \frac{\left| \bigcup_{i \in C} B_i \right|}{\left| \bigcup_{i \in N} B_i \right|}
    \]
    where \( C \) is the set of clique members and \( N \) is the set of all players. If the denominator is zero, the probability is defined to be 1.
    \item If a new knowledge item is created, it is added to the knowledge set \( B_i \) of each clique member \( i \in C \).
    \item If the selected clique contains any edge added by \texttt{tg}, then \texttt{tg} receives a reward of 1.
    \item The vertex, \(v_i\), with the highest knowledge, i.e., with \( |B_i|\) being maximum, leaves the organization in probability 

    \[
    P = \frac{\left|  B_i \right|}{\left| \bigcup_{i \in N} B_i \right|}
    \]

    If \( \left| \bigcup_{i \in N} B_i \right| = 0 \) then \(p = 0\).  If it is decided that the vertex, \(v_i\), with the highest knowledge leaves the organization then it is dropped from from the graph $G$ and the edges connected to \(v_i\) are dropped from the graph as well.  
\end{itemize}

\medskip

To evaluate contributions using the Shapley value, we define the cooperative game \( (N, v) \), where:
\[
v(S) = \left| \bigcup_{i \in S} B_i \right|
\]
for any coalition \( S \subseteq N \). The Shapley value of player \( i \in N \) is computed as:
\[
\phi_i = \sum_{j \in B_i} \frac{1}{|\{k \in N \mid j \in B_k\}|}
\]
This reflects the fractional contribution of each player to the collective knowledge pool, accounting for shared ownership.

\medskip

\textbf{Simulation Task one:}
\begin{itemize}
    \item Initialize a graph with five players:
    \begin{itemize}
        \item \( v_1 \) connected to \( v_2 \) and \( v_3 \)
        \item \( v_2 \) connected to \( v_3 \)
        \item \( v_3 \) connected to \( v_4 \) and \( v_5 \)
        \item \( v_4 \) connected to \( v_5 \)
    \end{itemize}
    \item Include a \texttt{tg} player with \( k = 1 \)
    \item Run the simulation for 10 stages
    \item In one simulation, let \texttt{tg} choose an edge randomly
    \item In another simulation, let \texttt{tg} always add the edge \( (v_1, v_4) \)
    \item Compute and plot:
    \begin{itemize}
        \item The cumulative reward of \texttt{tg} after each stage
        \item The knowledge growth of each player after each stage
        \item The Shapley value \( \phi_i \) of each player after each stage
    \end{itemize}
\end{itemize}

\textbf{Simulation Task two:}
\begin{itemize}
    \item Initialize a graph with ten players:
    \begin{itemize}
        \item \( v_1 \), \( v_2 \) to \( v_5 \) connected to each other.
        \item \( v_5 \), \( v_6 \) to  \( v_{10} \) connected to each other. 
    \end{itemize}
    \item Include a \texttt{tg} player with \( k = 5 \)
    \item Run the simulation for 10 stages
    \item In one simulation, let \texttt{tg} choose to connect 5 pairs of vertices that are not connected.  For example, \texttt{tg} may choose to connect \( v_1 \) with \( v_{10} \).
    \item In another simulation, let \texttt{tg} always connect \( v_4 \) with \( v_6 \) to. \( v_{10} \).
    \item Compute and plot:
    \begin{itemize}
        \item The cumulative reward of \texttt{tg} after each stage.
        \item The knowledge growth of each player after each stage.
        \item The Shapley value \( \phi_i \) of each player after each stage. 
    \end{itemize}
\end{itemize}

Report that a vertex left the organization and at which stage of the simulation.  Plot its knowledge growth up to the stage of the simulation in which it left the organization.

\end{humanprompt}

\subsection{Ambiguities related to leaving the organization}

\begin{humanprompt}

I used the following prompt - "I added to the following specification a feature of leaving the organization. Please point out ambiguities related to that feature. Produce the report as a latex subsection" 

\end{humanprompt}

\subsubsection{Ambiguities and Inconsistencies in the ``Leaving the Organization'' Feature}

\begin{aireply}

\textbf{AI reply:} \\
The addition of the feature allowing a vertex (inventor) to leave the organization introduces several ambiguities and potential inconsistencies in the specification. These issues should be clarified before implementing or simulating the system.

\begin{enumerate}
    \item \textbf{Ambiguity in selection of departing vertex:}
    \begin{itemize}
        \item The specification states that ``the vertex \(v_i\) with the highest knowledge leaves the organization in probability \( P = \frac{|B_i|}{|\bigcup_{i \in N} B_i|} \).''  
        \item It is not clear how ties are resolved if multiple vertices share the same maximum knowledge level \( |B_i| \).
        \item It is also unclear whether the probability is applied only to that single vertex or whether multiple vertices can potentially leave in the same stage.
    \end{itemize}

    \item \textbf{Ambiguity in timing of departure within a stage:}
    \begin{itemize}
        \item The order of operations within each stage is not explicitly defined. Specifically, it is ambiguous whether the departure decision occurs:
        \begin{enumerate}
            \item before temporary edges are added,
            \item after knowledge generation,
            \item or at the end of the stage after all updates.
        \end{enumerate}
        \item This timing affects which vertices participate in knowledge sharing and which are eligible for temporary edges.
    \end{itemize}

    \item \textbf{Ambiguity in the definition of the departure probability denominator:}
    \begin{itemize}
        \item The denominator \( |\bigcup_{i \in N} B_i| \) depends on the current set \( N \) of players.
        \item Once a vertex leaves, the value of \( N \) changes, which can alter subsequent probabilities. The specification does not specify whether the denominator should be computed \emph{before} or \emph{after} the removal of the vertex.
    \end{itemize}

    \item \textbf{Inconsistency with the definition of the Shapley value:}
    \begin{itemize}
        \item The Shapley value formula assumes a fixed set of players \( N \). When a vertex leaves, it is unclear whether:
        \begin{enumerate}
            \item the Shapley value should be recomputed using the reduced set \( N' = N \setminus \{v_i\} \), or
            \item the leaving vertex retains its previously computed Shapley value.
        \end{enumerate}
        \item This affects the interpretation of contributions over time and the comparability of Shapley values across stages.
    \end{itemize}

    \item \textbf{Ambiguity in knowledge persistence:}
    \begin{itemize}
        \item When a vertex leaves, it is unclear whether its knowledge set \( B_i \) is permanently removed or remains as part of the collective knowledge \( \bigcup_{i \in N} B_i \).
        \item If its knowledge is removed, the denominator in the probability formula and the cooperative game value function \( v(S) \) change significantly.
        \item If its knowledge remains, then the organization may continue to benefit from knowledge of a departed member, which may contradict the intended semantics of ``leaving.''
    \end{itemize}

    \item \textbf{Ambiguity in reporting and visualization:}
    \begin{itemize}
        \item The specification requires to ``report that a vertex left the organization and at which stage of the simulation,'' but it does not specify how this report integrates with cumulative plots.
        \item For example, it is unclear whether the knowledge growth plot should:
        \begin{enumerate}
            \item stop at the stage the vertex leaves,
            \item continue with zero growth thereafter,
            \item or exclude the vertex entirely from subsequent plots.
        \end{enumerate}
    \end{itemize}

    \item \textbf{Potential inconsistency in probabilistic interpretation:}
    \begin{itemize}
        \item The probability \( P = \frac{|B_i|}{|\bigcup_{i \in N} B_i|} \) can exceed 1 if multiple vertices share identical knowledge or if \( \bigcup_{i \in N} B_i \) is not properly updated.
        \item There is no normalization or upper bound enforcement described to ensure \( 0 \leq P \leq 1 \).
    \end{itemize}

    \item \textbf{Ambiguity in facilitator (\texttt{tg}) interaction:}
    \begin{itemize}
        \item The facilitator's reward depends on temporary edges, but if a vertex involved in a rewarded edge leaves in the same stage, it is unclear whether that reward still counts.
        \item The specification should clarify whether the facilitator’s reward is computed before or after departures.
    \end{itemize}
\end{enumerate}

\medskip

Clarifying these issues will ensure consistent simulation results and proper interpretation of the ``leaving the organization'' dynamics in both the knowledge and Shapley value computations.

\end{aireply}

\subsection{Clarification of the leave the organization feature}

\begin{humanprompt}

\textbf{Prompt:} \\
Incorporate the following clarification and an updated description as a latex subsection.  

\begin{enumerate}
    \item If there is more than one inventor with maximal knowledge randomly choose one of them to leave the organization.  Only one inventor leave the organization at a simulation stage. 
    \item Have the step of an inventor leaving the organization occur as the first step at each stage of the simulation. If in the first step of the stage an inventor \( v_i \) left the organization, All subsequent steps of a stage should be done on the updated graph, \( G \) and the knowledge sets associated only with the new set of inventors in \( G\) that does not include \( v_i \).  Specifically, the Shapely value should be calculated on the updated graph. 
    \item If $v_i$ leaves the organization, its set of knowledge items, $B_{v_i}$, is no longer available to the organization and is not used in the different calculations. 
    \item If \( v_i \) leaves the organization the plots should stop updating its associated different plots.  For example the knowledge plot should represents the growth of the knowledge of \( v_i \) up to the point that it left the organization. 
    
\end{enumerate}

\end{humanprompt}

\subsection{Revised Simulation Specification with the ``Leaving the Organization'' Feature}

\begin{aireply}

\textbf{AI reply:} \\
We are given a graph \( G \) whose vertices represent a set of inventors \( v_1, v_2, \ldots, v_n \). An edge \( (v_i, v_j) \) indicates that \( v_i \) and \( v_j \) can communicate. Each inventor \( v_i \) is associated with a set of knowledge items \( B_i \), initialized to the empty set.

\medskip

The game is played in discrete stages. In each stage, the following sequence of operations is executed:

\begin{enumerate}
    \item \textbf{Leaving the organization step:}
    \begin{itemize}
        \item Identify the inventor(s) with the maximal knowledge level, i.e., those for whom \( |B_i| \) is maximal.
        \item If there is more than one such inventor, randomly choose one inventor \( v_i \) among them.
        \item The inventor \( v_i \) leaves the organization with probability
        \[
        P = \frac{|B_i|}{\left| \bigcup_{j \in N} B_j \right|}.
        \]
        \item If \( \left| \bigcup_{j \in N} B_j \right| = 0 \), then \( P = 0 \).
        \item If \( v_i \) leaves the organization:
        \begin{itemize}
            \item Remove vertex \( v_i \) and all its incident edges from the graph \( G \).
            \item Remove its knowledge set \( B_{v_i} \) from all further calculations.
            \item The set of active inventors becomes \( N' = N \setminus \{v_i\} \).
            \item All subsequent operations in this stage are performed using the updated graph \( G' \) and inventor set \( N' \).
        \end{itemize}
        \item The departure is recorded in the simulation output, including the stage number and the identifier of the departing vertex. The plots for \( v_i \) (e.g., knowledge growth) are truncated at this stage.
    \end{itemize}

    \item \textbf{Facilitator edge addition:}
    \begin{itemize}
        \item A facilitator player, denoted \texttt{tg}, may add up to \( k \) temporary edges to the graph, where \( k < n^2 \).
        \item These edges exist only for the duration of the stage and are removed afterward.
    \end{itemize}

    \item \textbf{Clique selection and knowledge creation:}
    \begin{itemize}
        \item A single clique is selected randomly from the graph, including the temporary edges.
        \item A new knowledge item is created for the selected clique with probability:
        \[
        P = \frac{\left| \bigcup_{i \in C} B_i \right|}{\left| \bigcup_{i \in N'} B_i \right|},
        \]
        where \( C \) is the set of clique members and \( N' \) is the current set of active inventors.
        \item If the denominator is zero, the probability is defined to be 1.
        \item If a new knowledge item is created, it is added to the knowledge set \( B_i \) of each clique member \( i \in C \).
    \end{itemize}

    \item \textbf{Facilitator reward:}
    \begin{itemize}
        \item If the selected clique contains any temporary edge added by \texttt{tg}, then \texttt{tg} receives a reward of 1 for this stage.
    \end{itemize}

    \item \textbf{Shapley value computation:}
    \begin{itemize}
        \item To evaluate contributions using the Shapley value, define the cooperative game \( (N', v) \), where:
        \[
        v(S) = \left| \bigcup_{i \in S} B_i \right|, \quad \text{for any } S \subseteq N'.
        \]
        \item The Shapley value of player \( i \in N' \) is computed as:
        \[
        \phi_i = \sum_{j \in B_i} \frac{1}{|\{k \in N' \mid j \in B_k\}|}.
        \]
        \item This value reflects each inventor's fractional contribution to the collective knowledge pool, considering shared ownership.
    \end{itemize}
\end{enumerate}

\medskip

\textbf{Simulation Task One:}
\begin{itemize}
    \item Initialize a graph with five inventors:
    \begin{itemize}
        \item \( v_1 \) connected to \( v_2 \) and \( v_3 \)
        \item \( v_2 \) connected to \( v_3 \)
        \item \( v_3 \) connected to \( v_4 \) and \( v_5 \)
        \item \( v_4 \) connected to \( v_5 \)
    \end{itemize}
    \item Include a facilitator player \texttt{tg} with \( k = 1 \).
    \item Run the simulation for 10 stages.
    \item In one simulation, let \texttt{tg} choose an edge randomly.
    \item In another simulation, let \texttt{tg} always add the edge \( (v_1, v_4) \).
    \item Compute and plot:
    \begin{itemize}
        \item The cumulative reward of \texttt{tg} after each stage.
        \item The knowledge growth of each inventor after each stage.
        \item The Shapley value \( \phi_i \) of each inventor after each stage.
        \item The knowledge trajectory of any inventor who left, up to the stage of departure.
    \end{itemize}
\end{itemize}

\textbf{Simulation Task Two:}
\begin{itemize}
    \item Initialize a graph with ten inventors:
    \begin{itemize}
        \item \( v_1 \) to \( v_5 \) connected to each other.
        \item \( v_5 \) to \( v_{10} \) connected to each other.
    \end{itemize}
    \item Include a facilitator player \texttt{tg} with \( k = 5 \).
    \item Run the simulation for 10 stages.
    \item In one simulation, let \texttt{tg} choose to connect 5 random pairs of currently unconnected vertices (e.g., \( v_1 \) with \( v_{10} \)).
    \item In another simulation, let \texttt{tg} always connect \( v_4 \) with \( v_6, v_7, v_8, v_9, v_{10} \).
    \item Compute and plot:
    \begin{itemize}
        \item The cumulative reward of \texttt{tg} after each stage.
        \item The knowledge growth of each inventor after each stage.
        \item The Shapley value \( \phi_i \) of each inventor after each stage.
        \item The knowledge trajectories of inventors who leave the organization.
    \end{itemize}
\end{itemize}

\textbf{Added based on reivew of the simulation results.:} \\
The simulation is conducted on two graphs.  For each graph, the \( tg \) uses to strategies.  For each strategy we need to get 2 plots for the \( tg \) accumulative reward, accumulative knowledge per player and Shapely value of each player.  Thus, the specification produces 12 plots.

\textbf{Added based on reivew of the simulation results.:} \\
Update the program to include a diagram that for each tg strategy and simulation or graph we see the total amount of knowledge that was produce and the percentage that that left the organization.

\medskip

This revised specification provides an unambiguous, stepwise definition of the simulation process, ensuring consistency in the treatment of inventor departures, graph updates, and cooperative game computations.

\end{aireply}

\subsection{Result of the simulation}

The result of the simulation is given \href{https://colab.research.google.com/drive/1qFuYmyzGHcDNqN-rTAgYEA_L8LuBgEbL#scrollTo=ZDPxsD35tF3s&line=288&uniqifier=1}{here}.  

\textbf{Based on the simulation results.:} \\
Maybe I want inventor to leave at most once in a simulation.

\textbf{Based on the simulation results.:} \\
Adding a diagram of total knowledge produced and total knowledge that left the organization was easier when adding it to the specification and generating than when asking to add it to the existing code.  Adding to the existing code create code that failed at runtime in spite repeated attempts. 

Number of lines of code is now 443 including comments.



\section{More than one organization}

Next I'm attempting to enhance the simulation to simulate knowledge growth in more than one organization with inventors moving from one organization to another. 

\subsection{Simulation Specification with more than one organization}

\begin{humanprompt}

\textbf{Prompt:} \\
We are given \( k \) graphs \( G_1, \ldots, G_k \) whose vertices represent a set of inventors \( v^l_1, v^l_2, \ldots, v^l_{n_l} \) for the graph $G_l$. An edge \( (v^l_i, v^l_j) \) indicates that \( v^l_i \) and \( v^l_j \) can communicate. Each inventor \( v^l_i \) is associated with a set of knowledge items \( B^l_i \), initialized to the empty set.

\medskip

The game is played in discrete stages. In each stage, the following sequence of operations is executed:

\begin{enumerate}
    \item Randomly choose a graph $G_l$.
    \item \textbf{Leaving the organization step:}
    \begin{itemize}
        \item This stage is executed only if the number of graphs, \( k \), is greater than one.  
        \item Identify the inventor(s) with the maximal knowledge level, i.e., those for whom \( |B^l_i| \) is maximal.
        \item If there is more than one such inventor, randomly choose one inventor \( v^l_i \) among them.
        \item The inventor \( v^l_i \) leaves the organization with probability
        \[
        P = \frac{|B^l_i|}{\left| \bigcup_{j \in N_l} B^l_j \right|}.
        \]
        \item If \( \left| \bigcup_{j \in N_l} B_j \right| = 0 \), then \( P = 0 \).
        \item If \( v^l_i \) leaves the organization:
        \begin{itemize}
            \item Remove vertex \( v^l_i \) and all its incident edges from the graph \( G_l \).
            \item Remove its knowledge set \( B^l_{v^l_i} \) from all further calculations that has to do with \( G_l \).
            \item The set of active inventors in \(G_l\) becomes \( N_l' = N_l \setminus \{v^l_i\} \).
            \item All subsequent operations in this stage are performed using the updated graph \( G_l' \) and inventor set \( N_l' \).
            \item Randomly choose a graph \(G_r\) for which \( r \ne l\).  Add \(v^l_i\) to \( G_r\) with knowledge set \( B^l_i\).  That is \(v^l_i\) is moved to  \( G_r\) with its knowledge set.     
        \end{itemize}
        \item The departure is recorded in the simulation output, including the stage number and the identifier of the departing vertex. It is explicitly stated which graph the vertex left and which graph it joined. The plots for \( v^l_i \) (e.g., knowledge growth) continue to be updated but it is clearly marked in the plot to which graph the vertex belong to. 
    \end{itemize}

    \item \textbf{Facilitator edge addition:}
    \begin{itemize}
        \item A facilitator player, denoted \texttt{\(tg_l\)}, may add up to \( k \) temporary edges to the graph \(G_l\), where \( k < n^2 \).
        \item These edges exist only for the duration of the stage and are removed afterward.
    \end{itemize}

    \item \textbf{Clique selection and knowledge creation:}
    \begin{itemize}
        \item A single clique is selected randomly from the graph, including the temporary edges.
        \item A new knowledge item is created for the selected clique with probability:
        \[
        P = \frac{\left| \bigcup_{i \in C} B^l_i \right|}{\left| \bigcup_{i \in N^{l'}} B^l_i \right|},
        \]
        where \( C \) is the set of clique members and \( N^{l'} \) is the current set of active inventors.
        \item If the denominator is zero, the probability is defined to be 1.
        \item If a new knowledge item is created, it is added to the knowledge set \( B^l_i \) of each clique member \( i \in C \).
    \end{itemize}

    \item \textbf{Facilitator reward:}
    \begin{itemize}
        \item If the selected clique contains any temporary edge added by \texttt{tg}, then \texttt{tg} receives a reward of 1 for this stage.
    \end{itemize}

    \item \textbf{Shapley value computation:}
    \begin{itemize}
        \item To evaluate contributions using the Shapley value, define the cooperative game \( (N^{l'}, v) \), where:
        \[
        v(S) = \left| \bigcup_{i \in S} B^l_i \right|, \quad \text{for any } S \subseteq N^{l'}.
        \]
        \item The Shapley value of player \( i \in N^{l'} \) is computed as:
        \[
        \phi_i = \sum_{j \in B^l_i} \frac{1}{|\{k \in N^{l'} \mid j \in B^l_k\}|}.
        \]
        \item This value reflects each inventor's fractional contribution to the collective knowledge pool, considering shared ownership.
    \end{itemize}
\end{enumerate}

\medskip

\textbf{Simulation Task One:}
\begin{itemize}
    \item In this simulation the number of graphs is \( 1 \) 
    \item Initialize a graph with five inventors:
    \begin{itemize}
        \item \( v_1 \) connected to \( v_2 \) and \( v_3 \)
        \item \( v_2 \) connected to \( v_3 \)
        \item \( v_3 \) connected to \( v_4 \) and \( v_5 \)
        \item \( v_4 \) connected to \( v_5 \)
    \end{itemize}
    \item Include a facilitator player \texttt{tg} with \( k = 1 \).
    \item Run the simulation for 10 stages.
    \item In one simulation, let \texttt{tg} choose an edge randomly.
    \item In another simulation, let \texttt{tg} always add the edge \( (v_1, v_4) \).
    \item Compute and plot:
    \begin{itemize}
        \item The cumulative reward of \texttt{tg} after each stage.
        \item The knowledge growth of each inventor after each stage.
        \item The Shapley value \( \phi_i \) of each inventor after each stage.
        \item The knowledge trajectory of any inventor who left, up to the stage of departure.
    \end{itemize}
\end{itemize}

\textbf{Simulation Task Two:}
\begin{itemize}
    \item In this simulation the number of graphs is \( 1 \) 
    \item Initialize a graph with ten inventors:
    \begin{itemize}
        \item \( v_1 \) to \( v_5 \) connected to each other.
        \item \( v_5 \) to \( v_{10} \) connected to each other.
    \end{itemize}
    \item Include a facilitator player \texttt{tg} with \( k = 5 \).
    \item Run the simulation for 10 stages.
    \item In one simulation, let \texttt{tg} choose to connect 5 random pairs of currently unconnected vertices (e.g., \( v_1 \) with \( v_{10} \)).
    \item In another simulation, let \texttt{tg} always connect \( v_4 \) with \( v_6, v_7, v_8, v_9, v_{10} \).
    \item Compute and plot:
    \begin{itemize}
        \item The cumulative reward of \texttt{tg} after each stage.
        \item The knowledge growth of each inventor after each stage.
        \item The Shapley value \( \phi_i \) of each inventor after each stage.
        \item The knowledge trajectories of inventors who leave the organization.
    \end{itemize}
\end{itemize}

\textbf{Simulation Task Three:}

\begin{itemize}
    \item In this simulation the number of graphs is \( 2 \) 
    \begin{itemize}
        \item Initialize \( G_1 \) with ten inventors:
        \begin{itemize}
            \item \( v^1_1 \) to \( v^1_5 \) connected to each other.
            \item \( v^1_5 \) to \( v^1_{10} \) connected to each other.
        \end{itemize}
    \end{itemize}
    \begin{itemize}
        \item Initialize \( G_2 \) with ten inventors:
        \begin{itemize}
            \item \( v^2_1 \) to \( v^2_5 \) connected to each other.
            \item \( v^2_5 \) to \( v^2_{10} \) connected to each other.
        \end{itemize}
    \end{itemize}            
    \item Include a facilitator player \texttt{\(tg_1\)} with \( k = 5 \).  Let \texttt{\(tg_1\)} choose to connect 5 random pairs of currently unconnected vertices (e.g., \( v^1_1 \) with \( v^1_{10} \)).
    \item Include a facilitator, \texttt{\(tg_2\)} with \( k = 1 \).  Let \texttt{\(tg_2\)} choose to connect 1 random pairs of currently unconnected vertices (e.g., \( v^1_1 \) with \( v^1_{10} \)).
    \item Compute and plot:
    \begin{itemize}
        \item The cumulative reward of \texttt{\(tg_1\)} after each stage.
        \item The cumulative reward of \texttt{\(tg_2\)} after each stage.
        \item The knowledge growth of each inventor after each stage.
        \item The Shapley value \( \phi_i \) of each inventor after each stage.
    \end{itemize}
    \item Run the simulation for 10 stages.

\end{itemize}

Each plot should be clearly identified clearly associated with the simulation task and strategy used by the $tg_i$s.  In addition the plots should be both written to a file and standard output.

\
When there is more than one graph, create a table that represents what percentage of an inventor knowledge was created in each of the organizations. 

\end{humanprompt}

\subsection{Ambiguities and Inconsistencies in the Multi-Organization Specification}

\begin{aireply}

\textbf{AI reply to ambiguity query:} \\
The following points identify ambiguities and inconsistencies that arise in the specification when simulating more than one organization.

\subsubsection*{1. Ambiguity in Inventor Identity Across Organizations}
\begin{itemize}
    \item It is unclear whether an inventor who moves from organization \( G_l \) to \( G_r \) retains the same global identifier (e.g., \( v^l_i \)) or receives a new one (e.g., \( v^r_j \)).  
    \item The notation \( v^l_i \) suggests that the superscript indicates the organization, but when the inventor moves, the new notation \( v^r_i \) would imply a different entity, not the same inventor.  
    \item This ambiguity affects the interpretation of knowledge trajectories and the computation of cumulative statistics that track inventors across organizations.
\end{itemize}

\subsubsection*{2. Inconsistency in Knowledge Transfer Upon Movement}
\begin{itemize}
    \item When an inventor moves from \( G_l \) to \( G_r \), it is stated that the inventor carries its knowledge set \( B^l_i \).  
    \item However, it is not specified whether the inventor’s knowledge is:
    \begin{enumerate}
        \item merged into the new organization’s collective knowledge pool, or  
        \item kept separate until interactions occur in \( G_r \).  
    \end{enumerate}
    \item It is also unclear whether this transfer influences the knowledge denominator terms in later probability computations within \( G_r \).
\end{itemize}

\subsubsection*{3. Ambiguity in Facilitator Assignment}
\begin{itemize}
    \item The text introduces multiple facilitators (\texttt{tg\_1}, \texttt{tg\_2}) but does not clarify whether each facilitator operates only within its associated graph (\( G_1 \), \( G_2 \)) or may influence other graphs when inventors move between them.
    \item The notation “a facilitator player, denoted \texttt{\( tg_l \)}” suggests a per-organization facilitator, but this is inconsistent with later tasks where facilitators are independently parameterized.
\end{itemize}

\subsubsection*{4. Inconsistency in Probability Definition After Movement}
\begin{itemize}
    \item The probability of leaving an organization is defined as
    \[
    P = \frac{|B^l_i|}{\left| \bigcup_{j \in N_l} B^l_j \right|}.
    \]
    However, after the inventor leaves and joins \( G_r \), there is no clarification on how this probability behaves in future iterations:
    \begin{itemize}
        \item Does the inventor immediately become eligible to leave \( G_r \) again in the next stage?
        \item Are there any restrictions or cooldowns on movement?
    \end{itemize}
\end{itemize}

\subsubsection*{5. Ambiguity in the Stage Ordering and Global Stage Definition}
\begin{itemize}
    \item The specification says, “Randomly choose a graph \( G_l \)” at the start of each stage, implying that only one graph is active per stage.  
    \item It is unclear whether a “stage” is global (affecting all organizations) or local (executed independently per graph).  
    \item This affects the timing of departures, facilitator actions, and knowledge updates — especially when inventors move between organizations.
\end{itemize}

\subsubsection*{6. Inconsistency in the Use of the Facilitator Parameter \( k \)}
\begin{itemize}
    \item The facilitator parameter \( k \) is reused for different meanings:
    \begin{itemize}
        \item As the number of graphs \( (k) \).
        \item As the number of temporary edges a facilitator may add.
    \end{itemize}
    \item This overloading of notation can lead to confusion when interpreting the rules for multi-organization simulations.
\end{itemize}

\subsubsection*{7. Ambiguity in Global Knowledge Accounting}
\begin{itemize}
    \item The specification requests a table showing “what percentage of an inventor’s knowledge was created in each organization.”  
    \item However, there is no definition of how to attribute knowledge to an organization:
    \begin{itemize}
        \item Is knowledge attributed based on the organization where the knowledge item was first created?
        \item Or based on where the inventor was located when it was acquired?
    \end{itemize}
    \item This ambiguity affects both the tracking of knowledge evolution and the validity of comparative plots between organizations.
\end{itemize}

\subsubsection*{8. Missing Clarification on Re-Integration of Inventors}
\begin{itemize}
    \item It is not specified whether an inventor can return to its original organization after leaving.
    \item If this can happen, there is no description of how edges and prior relationships in \( G_l \) are restored (e.g., are previous connections reactivated or randomized?).
\end{itemize}

\subsubsection*{9. Inconsistency in Simulation Tasks}
\begin{itemize}
    \item In \textbf{Simulation Task Three}, the specification clearly intends to model two organizations, yet facilitators \texttt{\( tg_1 \)} and \texttt{\( tg_2 \)} both refer to operations within \( G_1 \) (not \( G_1 \) and \( G_2 \) respectively).  
    \item This contradicts the initial definition of having separate facilitators per organization and leaves unclear whether \( \texttt{tg}_2 \) operates in \( G_2 \).
\end{itemize}

\subsubsection*{10. Ambiguity in Plotting and Output Association}
\begin{itemize}
    \item The text requires that “each plot should be clearly associated with the simulation task and strategy used by the \( tg_i \)s,” but does not specify how inventors moving between organizations are to be represented in plots (e.g., continuous line with color change, or split plots).
    \item It is also unclear whether cumulative statistics (e.g., Shapley value evolution) are computed per organization or globally after movement.
\end{itemize}

\bigskip
\textbf{Summary:}  
The main inconsistencies center on unclear definitions of identity, ownership, and timing across organizations. To ensure an unambiguous simulation model, the specification should precisely define:
\begin{enumerate}
    \item Global inventor identity and how it changes upon movement.
    \item Whether stages are global or organization-specific.
    \item How knowledge and probabilities are recalculated after transfer.
    \item Which facilitator acts in which organization.
    \item How to attribute knowledge to organizations for reporting.
\end{enumerate}

\end{aireply}

\subsection{Clarifying ambiguity}

\begin{humanprompt}

\textbf{Ambiguity clarifying prompt:} \\

\begin{enumerate}
    \item An inventor is a entity that has a unique id.  When an inventor leaves an organization and joins an organization it retains its unique id.  
    \item An inventor $v$ has a list of knowledge items $B_v$ associated with her.  When she joins an organization or leaves an organization the association between $v$ and $B_v$ is retained.  If $v$ leaves $G_l$ and joins $G_r$ then $G_r$ now has a new inventor $v$ who owns knowledge items $B_v$.  Next ins any simulation calculations that include $v$ as a member of the organization $G_r$ its associated knowledge item $B_v$ should be taken into account and used appropriately in the calculation of $v(S)$, etc.  Thus, for example if $v \in S$, $v(S) = |\bigcup_{i \in S} B_i |$  will include $B_v$ as well. 
    \item A facilitator $tg_i$ is associated with an organization $G_i$.  As a consequence, $tg_i$ can only add edges to the graph $G_i$.
    \item In a stage, leaving the organization step is preformed in probability $p$, for now have $p = 0.2$
    \item Each stage is preformed on a randomly chosen graph $G_l$.  If within the stage an inventor $v$ leaves $G_l$ and joins $G_r$ then the graph $G_r$ is impacted as well.  In such as case $G_r$ would have another investor $v$ with its associated knowledge items $B_v$. 
    \item Knowledge is attributed based on the organization where the knowledge item
was first created.
    \item An inventor can go back to an organization.   To keep it simple for now no edges are avaiaible when the inventor rejoins the organization.
    \item In simulation task three $tg_1$ should work on $G_1$ and $tg_2$ on $G_2$.
    \item Inventors that moved organization should be represented in the plots by continuous line with
color change.
    \item Cumulative statistics should be calculated globally after each stage. 
\end{enumerate}

\end{humanprompt}

\subsection{Multi-Organization Simulation Specification}

\begin{humanprompt}

\textbf{Prompt to obtain an updated specification:} \\
here you are given a specification of a simulation, your ambiguity report and clarification by me of the ambiguities. Please create a new latex subsection specification of the simulation that incorporate the clarifications.  

\end{humanprompt}

\textbf{After some clarification questions the AI produced the following updated specification. } \\

\begin{aireply}

This subsection defines a simulation model for knowledge creation and diffusion across multiple organizations. Each organization is represented as a graph, and inventors may move between organizations over time. The goal is to study how network structure, mobility, and facilitator interventions influence knowledge growth and individual contribution.

\subsubsection*{Model Entities and Notation}

\paragraph{Organizations and Inventors.}
There are \(m \ge 1\) organizations, each represented as a graph \(G_l = (N_l, E_l)\) for \(l = 1, \ldots, m\). An inventor is an individual with a \textbf{global unique identifier}. If inventor \(i\) is currently in organization \(l\), we denote this inventor by \(v_i^{(l)}\). The subscript \(i\) never changes, even if the inventor moves between organizations; the superscript indicates the current organization membership.

\paragraph{Knowledge Sets.}
Each inventor \(v_i^{(l)}\) owns a set of knowledge items \(B_i\) (which persists with the inventor globally, regardless of organization). The set of knowledge in organization \(G_l\) at a given time is defined as
\[
K_l = \bigcup_{v_i^{(l)} \in N_l} B_i .
\]

\paragraph{Knowledge Attribution.}
Each knowledge item is attributed to the organization in which it was first created. If knowledge item \(x\) was generated in organization \(G_l\), its origin is recorded as \(\text{origin}(x) = l\). This information will later be used to compute, for each inventor, the percentage of their knowledge originating from each organization.

\paragraph{Facilitators.}
Each organization \(G_l\) has an associated facilitator denoted \texttt{tg\(_l\)}. Facilitator \texttt{tg\(_l\)} may add temporary edges to \(G_l\) only; it cannot modify other organizations.

\paragraph{Global Stage.}
Time evolves in discrete \textbf{global stages}. In each stage, exactly one organization is selected uniformly at random to execute the stage operations. Some operations (e.g., inventor movement) may affect multiple organizations.

\subsubsection*{Stage Operations}

Each global stage consists of the following sequence:

\begin{enumerate}
    \item \textbf{Organization Selection.}  
    Select an organization \(G_l\) uniformly at random from \(\{G_1, \ldots, G_m\}\). All subsequent steps in the stage apply with respect to \(G_l\).

    \item \textbf{Leaving-Organization Step (with probability \(p = 0.2\)).}
    This step is executed only if \(m > 1\).
    \begin{enumerate}
        \item Identify inventor(s) in \(G_l\) with maximal knowledge size:
        \[
        M_l = \left\{ v_i^{(l)} \in N_l \mid |B_i| = \max_{v_j^{(l)} \in N_l} |B_j| \right\}.
        \]
        \item If \(|M_l| > 1\), select one inventor uniformly at random from \(M_l\); denote this inventor by \(v_i^{(l)}\).
        \item Compute the probability of leaving:
        \[
        P_{\text{leave}} = 
        \begin{cases}
        \dfrac{|B_i|}{|K_l|}, & \text{if } |K_l| > 0, \\
        0, & \text{if } |K_l| = 0 .
        \end{cases}
        \]
        \item With probability \(P_{\text{leave}}\), the inventor leaves \(G_l\):
        \begin{itemize}
            \item Remove \(v_i^{(l)}\) and all incident edges from \(G_l\):  
            \( N_l \leftarrow N_l \setminus \{v_i^{(l)}\} \).
            \item Select an organization \(G_r\), \(r \ne l\), uniformly at random.
            \item Add the inventor to \(G_r\) with no edges:  
            \( N_r \leftarrow N_r \cup \{v_i^{(r)}\} \).  
            The inventor retains the same ID \(i\) and knowledge set \(B_i\).
            \item If an inventor returns to a previously visited organization, they rejoin with no pre-existing edges.
        \end{itemize}
        \item The movement event (departure and destination) is recorded for output and plotting.
    \end{enumerate}

    \item \textbf{Facilitator Temporary Edge Addition.}
    Facilitator \texttt{tg\(_l\)} may add up to \(k_l\) temporary edges to \(G_l\), where \(k_l \ge 0\) is a parameter of the simulation. All temporary edges are removed at the end of the stage.

    \item \textbf{Clique Selection and Knowledge Creation.}
    \begin{enumerate}
        \item A clique \(C \subseteq N_l\) is selected uniformly at random from all cliques in \(G_l\), considering both permanent and temporary edges.
        \item A new knowledge item is created for clique \(C\) with probability:
        \[
        P_{\text{new}} = 
        \begin{cases}
        \dfrac{\left|\bigcup_{v_i^{(l)} \in C} B_i\right|}{|K_l|}, & \text{if } |K_l| > 0, \\
        1, & \text{if } |K_l| = 0 .
        \end{cases}
        \]
        \item If a new item \(x\) is created:
        \begin{itemize}
            \item Add \(x\) to each \(B_i\) for all \(v_i^{(l)} \in C\).
            \item Record \(\text{origin}(x) = l\).
        \end{itemize}
    \end{enumerate}

    \item \textbf{Facilitator Reward.}
    Facilitator \texttt{tg\(_l\)} receives a reward of 1 for the stage if the selected clique contains at least one temporary edge added by \texttt{tg\(_l\)}.

    \item \textbf{Shapley Value Computation.}
    For organization \(G_l\), define the cooperative game
    \[
    (N_l, v), \qquad v(S) = \left|\bigcup_{v_i^{(l)} \in S} B_i\right| \quad \forall\, S \subseteq N_l.
    \]
    The Shapley value of an inventor \(v_i^{(l)}\) is computed as:
    \[
    \phi_i^{(l)} = \sum_{x \in B_i} \frac{1}{|\{v_j^{(l)} \in N_l : x \in B_j\}|}.
    \]

    \item \textbf{Global Statistics Update.}
    After every stage:
    \begin{itemize}
        \item Global cumulative statistics for all inventors are updated.
        \item Plots for inventors that changed organizations show continuous curves, with color changes marking organization transitions.
    \end{itemize}
\end{enumerate}

\subsubsection*{Reporting Requirements}

At the end of the simulation:
\begin{itemize}
    \item All plots must clearly indicate the simulation task and facilitator strategy.
    \item Plots must be written to file and printed to standard output.
    \item A table must be generated for each inventor, reporting the percentage of their knowledge originating from each organization.
\end{itemize}

\end{aireply} 

\subsection{Generated simulation}

Reviewed the corrected specification and it looked cleaner.  Created simulation based on the unambiguous specification above is given here \href{https://colab.research.google.com/drive/1qFuYmyzGHcDNqN-rTAgYEA_L8LuBgEbL#scrollTo=NpGt5920KSF7&line=274&uniqifier=1}{link}.  Organizations were created by the AI.  In fact, the AI dropped the three simulations tasks for the above updated unambiguous specification that I have to bring back to the specification.  The AI also only printed the plots to a file and did not print them to standard output although the specification instructed it to do so. Number of lines of code is now $395$.  


\subsection{More than one organization - adding back simulation tasks}
\label{twoOrgSpec}

\begin{humanprompt}

\textbf{I attempted to create a clean specification adding back the simulation tasks that where dropped. } \\ 

This subsection defines a simulation model for knowledge creation and diffusion across multiple organizations. Each organization is represented as a graph, and inventors may move between organizations over time. The goal is to study how network structure, mobility, and facilitator interventions influence knowledge growth and individual contribution.

\subsubsection*{Model Entities and Notation}

\paragraph{Organizations and Inventors.}
There are \(m \ge 1\) organizations, each represented as a graph \(G_l = (N_l, E_l)\) for \(l = 1, \ldots, m\). An inventor is an individual with a \textbf{global unique identifier}. If inventor \(i\) is currently in organization \(l\), we denote this inventor by \(v_i^{(l)}\). The subscript \(i\) never changes, even if the inventor moves between organizations; the superscript indicates the current organization membership.

\paragraph{Knowledge Sets.}
Each inventor \(v_i^{(l)}\) owns a set of knowledge items \(B_i\) (which persists with the inventor globally, regardless of organization). The set of knowledge in organization \(G_l\) at a given time is defined as
\[
K_l = \bigcup_{v_i^{(l)} \in N_l} B_i .
\]

\paragraph{Knowledge Attribution.}
Each knowledge item is attributed to the organization in which it was first created. If knowledge item \(x\) was generated in organization \(G_l\), its origin is recorded as \(\text{origin}(x) = l\). This information will later be used to compute, for each inventor, the percentage of their knowledge originating from each organization.

\paragraph{Facilitators.}
Each organization \(G_l\) has an associated facilitator denoted \texttt{tg\(_l\)}. Facilitator \texttt{tg\(_l\)} may add temporary edges to \(G_l\) only; it cannot modify other organizations.

\paragraph{Global Stage.}
Time evolves in discrete \textbf{global stages}. In each stage, exactly one organization is selected uniformly at random to execute the stage operations. Some operations (e.g., inventor movement) may affect multiple organizations.

\subsubsection*{Stage Operations}

Each global stage consists of the following sequence:

\begin{enumerate}
    \item \textbf{Organization Selection.}  
    Select an organization \(G_l\) uniformly at random from \(\{G_1, \ldots, G_m\}\). All subsequent steps in the stage apply with respect to \(G_l\).

    \item \textbf{Leaving-Organization Step (with probability \(p = 0.2\)).}
    This step is executed only if \(m > 1\).
    \begin{enumerate}
        \item Identify inventor(s) in \(G_l\) with maximal knowledge size:
        \[
        M_l = \left\{ v_i^{(l)} \in N_l \mid |B_i| = \max_{v_j^{(l)} \in N_l} |B_j| \right\}.
        \]
        \item If \(|M_l| > 1\), select one inventor uniformly at random from \(M_l\); denote this inventor by \(v_i^{(l)}\).
        \item Compute the probability of leaving:
        \[
        P_{\text{leave}} = 
        \begin{cases}
        \dfrac{|B_i|}{|K_l|}, & \text{if } |K_l| > 0, \\
        0, & \text{if } |K_l| = 0 .
        \end{cases}
        \]
        \item With probability \(P_{\text{leave}}\), the inventor leaves \(G_l\):
        \begin{itemize}
            \item Remove \(v_i^{(l)}\) and all incident edges from \(G_l\):  
            \( N_l \leftarrow N_l \setminus \{v_i^{(l)}\} \).
            \item Select an organization \(G_r\), \(r \ne l\), uniformly at random.
            \item Add the inventor to \(G_r\) with no edges:  
            \( N_r \leftarrow N_r \cup \{v_i^{(r)}\} \).  
            The inventor retains the same ID \(i\) and knowledge set \(B_i\).
            \item If an inventor returns to a previously visited organization, they rejoin with no pre-existing edges.
        \end{itemize}
        \item The movement event (departure and destination) is recorded for output and plotting.
    \end{enumerate}

    \item \textbf{Facilitator Temporary Edge Addition.}
    Facilitator \texttt{tg\(_l\)} may add up to \(k_l\) temporary edges to \(G_l\), where \(k_l \ge 0\) is a parameter of the simulation. All temporary edges are removed at the end of the stage.

    \item \textbf{Clique Selection and Knowledge Creation.}
    \begin{enumerate}
        \item A clique \(C \subseteq N_l\) is selected uniformly at random from all cliques in \(G_l\), considering both permanent and temporary edges.
        \item A new knowledge item is created for clique \(C\) with probability:
        \[
        P_{\text{new}} = 
        \begin{cases}
        \dfrac{\left|\bigcup_{v_i^{(l)} \in C} B_i\right|}{|K_l|}, & \text{if } |K_l| > 0, \\
        1, & \text{if } |K_l| = 0 .
        \end{cases}
        \]
        \item If a new item \(x\) is created:
        \begin{itemize}
            \item Add \(x\) to each \(B_i\) for all \(v_i^{(l)} \in C\).
            \item Record \(\text{origin}(x) = l\).
        \end{itemize}
    \end{enumerate}

    \item \textbf{Facilitator Reward.}
    Facilitator \texttt{tg\(_l\)} receives a reward of 1 for the stage if the selected clique contains at least one temporary edge added by \texttt{tg\(_l\)}.

    \item \textbf{Shapley Value Computation.}
    For organization \(G_l\), define the cooperative game
    \[
    (N_l, v), \qquad v(S) = \left|\bigcup_{v_i^{(l)} \in S} B_i\right| \quad \forall\, S \subseteq N_l.
    \]
    The Shapley value of an inventor \(v_i^{(l)}\) is computed as:
    \[
    \phi_i^{(l)} = \sum_{x \in B_i} \frac{1}{|\{v_j^{(l)} \in N_l : x \in B_j\}|}.
    \]

    \item \textbf{Global Statistics Update.}
    After every stage:
    \begin{itemize}
        \item Global cumulative statistics for all inventors are updated.
        \item Plots for inventors that changed organizations show continuous curves, with color changes marking organization transitions.
    \end{itemize}
\end{enumerate}

\textbf{Simulation Task One:}
\begin{itemize}
    \item In this simulation the number of graphs is \( 1 \) 
    \item Initialize a graph with five inventors:
    \begin{itemize}
        \item \( v_1 \) connected to \( v_2 \) and \( v_3 \)
        \item \( v_2 \) connected to \( v_3 \)
        \item \( v_3 \) connected to \( v_4 \) and \( v_5 \)
        \item \( v_4 \) connected to \( v_5 \)
    \end{itemize}
    \item Include a facilitator player \texttt{tg} with \( k = 1 \).
    \item Run the simulation for 10 stages.
    \item In one simulation, let \texttt{tg} choose an edge randomly.
    \item In another simulation, let \texttt{tg} always add the edge \( (v_1, v_4) \).
    \item Compute and plot:
    \begin{itemize}
        \item The cumulative reward of \texttt{tg} after each stage.
        \item The knowledge growth of each inventor after each stage.
        \item The Shapley value \( \phi_i \) of each inventor after each stage.

    \end{itemize}
\end{itemize}

\textbf{Simulation Task Two:}
\begin{itemize}
    \item In this simulation the number of graphs is \( 1 \) 
    \item Initialize a graph with ten inventors:
    \begin{itemize}
        \item \( v_1 \) to \( v_5 \) connected to each other.
        \item \( v_5 \) to \( v_{10} \) connected to each other.
    \end{itemize}
    \item Include a facilitator player \texttt{tg} with \( k = 5 \).
    \item Run the simulation for 10 stages.
    \item In one simulation, let \texttt{tg} choose to connect 5 random pairs of currently unconnected vertices (e.g., \( v_1 \) with \( v_{10} \)).
    \item In another simulation, let \texttt{tg} always connect \( v_4 \) with \( v_6, v_7, v_8, v_9, v_{10} \).
    \item Compute and plot:
    \begin{itemize}
        \item The cumulative reward of \texttt{tg} after each stage.
        \item The knowledge growth of each inventor after each stage.
        \item The Shapley value \( \phi_i \) of each inventor after each stage.

    \end{itemize}
\end{itemize}

\textbf{Simulation Task Three:}

\begin{itemize}
    \item In this simulation the number of graphs is \( 2 \) 
    \begin{itemize}
        \item Initialize \( G_1 \) with ten inventors:
        \begin{itemize}
            \item \( v^1_1 \) to \( v^1_5 \) connected to each other.
            \item \( v^1_5 \) to \( v^1_{10} \) connected to each other.
        \end{itemize}
    \end{itemize}
    \begin{itemize}
        \item Initialize \( G_2 \) with ten inventors:
        \begin{itemize}
            \item \( v^2_1 \) to \( v^2_5 \) connected to each other.
            \item \( v^2_5 \) to \( v^2_{10} \) connected to each other.
        \end{itemize}
    \end{itemize}            
    \item Include a facilitator player \texttt{\(tg_1\)} with \( k = 5 \).  Let \texttt{\(tg_1\)} choose to connect 5 random pairs of currently unconnected vertices (e.g., \( v^1_1 \) with \( v^1_{10} \)).
    \item Include a facilitator, \texttt{\(tg_2\)} with \( k = 1 \).  Let \texttt{\(tg_2\)} choose to connect 1 random pairs of currently unconnected vertices (e.g., \( v^1_1 \) with \( v^1_{10} \)).
    \item Compute and plot:
    \begin{itemize}
        \item The cumulative reward of \texttt{\(tg_1\)} after each stage.
        \item The cumulative reward of \texttt{\(tg_2\)} after each stage.
        \item The knowledge growth of each inventor after each stage.

    \end{itemize}
    \item Run the simulation for 10 stages.

\end{itemize}

Each plot should be clearly identified clearly associated with the simulation task and strategy used by the $tg_i$s.  In addition the plots should be both written to a file and standard output.

\
When there is more than one graph, create a table that represents what percentage of an inventor knowledge was created in each of the organizations. 

\end{humanprompt}




\section{Adding a correctness claim}

In this section we add a  simulation specification correctness claim.  We then attempt to integrate it with the current simulation implementation. 

\subsection{Correctness claim}

\begin{humanprompt}

The claim below is true after each stage of the simulation.  Assume that in a stage of the simulation the graph $G$ was chosen.  In addition, assume that the clique $C$ was chosen.
The $tg$ associated with this stage either added or did not edges connecting vertices in $C$.  

We will have a counter, $noTGKnowldge$.  In the beginning of the simulation $noTGKnowldge$ is set to $0$.

Next, if $C$ produces a new knowledge item, there are two cases to consider. 

\begin{enumerate}
    \item There are no edges that connect vertices in $C$ and where added by $tp$.
    In that case $noTGKnowldge$ stays the same.
    \item There are edges that connect vertices in $C$ and where added by $tg$.  In that case $noTGKnowldge$ is increment by $1$.  
\end{enumerate}

The second case is that $C$ does not produce a new knowledge item.  In that case  $noTGKnowldge$  is unchanged. 

We denote by $r(tg)$ the overall reword obtained by $tg$.  $r(tg)$ is initialized to $0$ in the beginning of the simulation and is updated after each stage of the simulation.  It is incremented by one if $tg$ gained a reward in the stage.

As previously defined, $B_{v}$ is the set of knowledge items the inventor $v$ has.  It is initialized to the empty set and is updated in each stage of the simulation.  

\begin{lemma}
In the initial stage of the system and after each stage of the simulation we have that 

$|\bigcup B_v| = noTGKnowldge + \sum r(tg)$

\end{lemma}

\end{humanprompt}

\subsection{Correctness proof}

\begin{humanprompt}
\textbf{I prompted the AI as follows.   
Here is a specification of simulation with a correctness claim. Please write a proof by induction on each stage of the claim. 
}
\end{humanprompt}

I also provided the current specification and the above claim.  I got the following updated claim and proof.

\begin{aireply}

\begin{lemma}[Correctness Invariant]
\label{lemma:correctness}
Let $B_v$ denote the knowledge set of inventor $v$.
After every stage of the simulation, we have
\[
\left| \bigcup_v B_v \right|
\;=\;
noTGKnowldge \;+\; \sum_{\text{all facilitators } tg} r(tg),
\]
where:
\begin{itemize}
    \item $noTGKnowldge$ counts the number of knowledge items created in stages where the selected clique contained \emph{no} temporary edge added by any facilitator.
    \item For each facilitator $tg$, the value $r(tg)$ counts the number of knowledge items created in stages where the selected clique contained at least one temporary edge added by that $tg$.
\end{itemize}
Hence, each created knowledge item is counted exactly once by exactly one of the two terms.
\end{lemma}

\begin{proof}
We proceed by induction on the number of completed stages.

\textbf{Base Case ($T = 0$).}
At the start of the simulation, no stages have been executed.  
All inventors have empty knowledge sets, hence
\[
\left|\bigcup_v B_v\right| = 0.
\]
By initialization, $noTGKnowldge = 0$ and $r(tg) = 0$ for all facilitators.
Therefore the invariant holds:
\[
0 = 0 + 0.
\]

\textbf{Inductive Step.}
Assume that after $T$ stages the equality holds:
\[
\left|\bigcup_v B_v\right| = noTGKnowldge + \sum_{tg} r(tg).
\]
We show that it remains true after stage $T{+}1$.

During stage $T{+}1$, one organization $G$ is selected, and a clique $C \subseteq G$ is chosen.
At most one new knowledge item $x$ may be created in this stage.

\medskip
\noindent
\emph{Case 1: No new knowledge item is created.}
Then no inventor's knowledge set changes, and the counters remain unchanged:
\[
\left|\bigcup_v B_v\right|, \; noTGKnowldge, \; \text{and each } r(tg)
\text{ all remain the same.}
\]
Thus the equality still holds after this stage.

\medskip
\noindent
\emph{Case 2: A new knowledge item is created.}
Then exactly one of the following subcases applies.

\begin{itemize}
    \item \textbf{Subcase (a):} 
    The chosen clique $C$ contains \emph{no} temporary edge added by any facilitator.
    Then the new item is counted as a ``non-tg'' item.
    Hence:
    \[
    \begin{aligned}
    \left|\bigcup_v B_v\right| &\gets \left|\bigcup_v B_v\right| + 1,\\
    noTGKnowldge &\gets noTGKnowldge + 1,\\
    r(tg) &\text{unchanged for all } tg.
    \end{aligned}
    \]
    Using the inductive hypothesis, we obtain:
    \[
    \left|\bigcup_v B_v\right|_{\!T{+}1} 
    = \left|\bigcup_v B_v\right|_{\!T} + 1
    = (noTGKnowldge + \sum r(tg)) + 1
    = (noTGKnowldge + 1) + \sum r(tg),
    \]
    which is exactly the required invariant.

    \item \textbf{Subcase (b):} 
    The chosen clique $C$ contains at least one temporary edge added by some facilitator $tg^\ast$.
    Then the new item is classified as a ``tg-involved'' item and increments $r(tg^\ast)$:
    \[
    \begin{aligned}
    \left|\bigcup_v B_v\right| &\gets \left|\bigcup_v B_v\right| + 1,\\
    noTGKnowldge &\text{unchanged},\\
    r(tg^\ast) &\gets r(tg^\ast) + 1.
    \end{aligned}
    \]
    Again, by the inductive hypothesis,
    \[
    \left|\bigcup_v B_v\right|_{\!T{+}1} 
    = \left|\bigcup_v B_v\right|_{\!T} + 1
    = (noTGKnowldge + \sum r(tg)) + 1
    = noTGKnowldge + \sum r(tg) + 1
    = noTGKnowldge_{\!T{+}1} + \sum r(tg)_{\!T{+}1}.
    \]
\end{itemize}

In all possible subcases, the invariant holds after stage $T{+}1$.

\medskip
\noindent
\textbf{Conclusion.}
By induction on the number of stages, the equality
\[
\left|\bigcup_v B_v\right| = noTGKnowldge + \sum_{tg} r(tg)
\]
holds after every stage of the simulation.
\end{proof}

\end{aireply}

\subsection{Proof review} 

I reviewed the proof above and it looks correct.  Next, I want to add an invariant to the implementation to check that the above lemma holds in the implementation, but first I asked the AI to suggest invariants.  Interestingly, the claim I defined above was not suggested. 


\subsection{Asking the AI to suggest invariants}

\begin{humanprompt}

\paragraph{Prompt.}
Please review the following specification and suggest invariants that apply to the specification and are critical for the correctness of the specification. (I provided the current specification above without my claim or the tasks (see \ref{twoOrgSpec}).)

\end{humanprompt}

\begin{aireply}

\paragraph{Response.}
Here are some critical invariants that should hold throughout the execution of the specification to ensure correctness and consistency:

\begin{enumerate}
    \item \textbf{Inventor Identity and Knowledge Persistence}
    \begin{itemize}
        \item \textbf{Invariant:} For every inventor \(i\), the global identifier remains constant across all organizations and stages:
        \[
        \forall i,\; \text{ID}(v_i^{(l)}) = i \text{ for any } l.
        \]
        \item \textbf{Invariant:} The knowledge set \(B_i\) of inventor \(i\) persists globally and is never lost or duplicated:
        \[
        \forall i,\; B_i \text{ is unique and consistent across organizations.}
        \]
    \end{itemize}

    \item \textbf{Organization Graph Integrity}
    \begin{itemize}
        \item \textbf{Invariant:} Each organization graph \(G_l = (N_l, E_l)\) satisfies:
        \[
        N_l \cap N_r = \emptyset \text{ for } l \neq r.
        \]
        \item \textbf{Invariant:} Temporary edges added by facilitator \(\texttt{tg}_l\) are removed at the end of the stage:
        \[
        E_l^{\text{temp}} = \emptyset \text{ after each stage.}
        \]
    \end{itemize}

    \item \textbf{Knowledge Attribution}
    \begin{itemize}
        \item \textbf{Invariant:} For every knowledge item \(x\), its origin is immutable:
        \[
        \text{origin}(x) \text{ is assigned once and never changes.}
        \]
    \end{itemize}

    \item \textbf{Probability Computations}
    \begin{itemize}
        \item \textbf{Invariant:} Leaving probability and new knowledge probability are always in \([0,1]\):
        \[
        0 \le P_{\text{leave}}, P_{\text{new}} \le 1.
        \]
        \item \textbf{Invariant:} If \(|K_l| = 0\), then \(P_{\text{leave}} = 0\) and \(P_{\text{new}} = 1\).
    \end{itemize}

    \item \textbf{Movement Consistency}
    \begin{itemize}
        \item \textbf{Invariant:} When an inventor moves:
        \begin{enumerate}
            \item They are removed from the source organization and added to the destination with no edges.
            \item Their knowledge set \(B_i\) remains unchanged.
            \item No duplicate inventor nodes are created.
        \end{enumerate}
    \end{itemize}

    \item \textbf{Clique Selection}
    \begin{itemize}
        \item \textbf{Invariant:} The selected clique \(C\) is a valid clique in \(G_l\) considering both permanent and temporary edges.
    \end{itemize}

    \item \textbf{Shapley Value Calculation}
    \begin{itemize}
        \item \textbf{Invariant:} For each inventor \(v_i^{(l)}\),
        \[
        \phi_i^{(l)} = \sum_{x \in B_i} \frac{1}{|\{v_j^{(l)} \in N_l : x \in B_j\}|}
        \]
        must be computed using the current organization membership only.
    \end{itemize}

    \item \textbf{Global Statistics}
    \begin{itemize}
        \item \textbf{Invariant:} Global statistics reflect cumulative knowledge and movements without contradiction (e.g., no negative counts, no missing transitions).
    \end{itemize}
\end{enumerate}

\end{aireply}

\subsection{Comments on AI suggested invariants}

Notably the invariant I defined was not suggested by the AI.  In addition, some of suggested invariants are more concrete and can be easily be proven while others are more abstract written in pure English with no logical formula associated with them. 




\subsection{Implementation with the invariant}

Created a clean specification with my invariant (invariant - \ref{lemma:correctness} specification -  \ref{promptGeneratingInvariant} and generated the code \href{https://colab.research.google.com/drive/1fEnkwHwHmAf5N6ItGj40mLMqeAYRLwqn#scrollTo=iGtMLsiLGwRf\&line=637&uniqifier=1}{link}. The code is now generating the three experiments with the two $tg$ strategies and is checking the invariant after each stage of the simulation.  The execution completes and the invariant is correct after each stage of the simulation. 


\section{Comparing two tg strategies that use $k$ edges}
\label{comparing_strategies}

Next I attempt to compare two tg strategies that use $k > 1$ edges as detailed in the prompt below.  

\begin{humanprompt}

\textbf{Please generate a python program that implements the following simulation specification and its associated invariant.}\\

\subsection{Comparing two tg strategies that use $k$ edges}
\label{twoOrgSpec}

This subsection defines a simulation model for knowledge creation and diffusion across multiple organizations. Each organization is represented as a graph, and inventors may move between organizations over time. The goal is to study how network structure, mobility, and facilitator interventions influence knowledge growth and individual contribution.

\subsubsection*{Model Entities and Notation}

\paragraph{Organizations and Inventors.}
There are \(m \ge 1\) organizations, each represented as a graph \(G_l = (N_l, E_l)\) for \(l = 1, \ldots, m\). An inventor is an individual with a \textbf{global unique identifier}. If inventor \(i\) is currently in organization \(l\), we denote this inventor by \(v_i^{(l)}\). The subscript \(i\) never changes, even if the inventor moves between organizations; the superscript indicates the current organization membership.

\paragraph{Knowledge Sets.}
Each inventor \(v_i^{(l)}\) owns a set of knowledge items \(B_i\) (which persists with the inventor globally, regardless of organization). The set of knowledge in organization \(G_l\) at a given time is defined as
\[
K_l = \bigcup_{v_i^{(l)} \in N_l} B_i .
\]

\paragraph{Knowledge Attribution.}
Each knowledge item is attributed to the organization in which it was first created. If knowledge item \(x\) was generated in organization \(G_l\), its origin is recorded as \(\text{origin}(x) = l\). This information will later be used to compute, for each inventor, the percentage of their knowledge originating from each organization.

\paragraph{Facilitators.}
Each organization \(G_l\) has an associated facilitator denoted \texttt{tg\(_l\)}. Facilitator \texttt{tg\(_l\)} may add temporary edges to \(G_l\) only; it cannot modify other organizations.

\paragraph{Global Stage.}
Time evolves in discrete \textbf{global stages}. In each stage, exactly one organization is selected uniformly at random to execute the stage operations. Some operations (e.g., inventor movement) may affect multiple organizations.

\subsubsection*{Stage Operations}

Each global stage consists of the following sequence:

\begin{enumerate}
    \item \textbf{Organization Selection.}  
    Select an organization \(G_l\) uniformly at random from \(\{G_1, \ldots, G_m\}\). All subsequent steps in the stage apply with respect to \(G_l\).

    \item \textbf{Leaving-Organization Step (with probability \(p = 0.2\)).}
    This step is executed only if \(m > 1\).
    \begin{enumerate}
        \item Identify inventor(s) in \(G_l\) with maximal knowledge size:
        \[
        M_l = \left\{ v_i^{(l)} \in N_l \mid |B_i| = \max_{v_j^{(l)} \in N_l} |B_j| \right\}.
        \]
        \item If \(|M_l| > 1\), select one inventor uniformly at random from \(M_l\); denote this inventor by \(v_i^{(l)}\).
        \item Compute the probability of leaving:
        \[
        P_{\text{leave}} = 
        \begin{cases}
        \dfrac{|B_i|}{|K_l|}, & \text{if } |K_l| > 0, \\
        0, & \text{if } |K_l| = 0 .
        \end{cases}
        \]
        \item With probability \(P_{\text{leave}}\), the inventor leaves \(G_l\):
        \begin{itemize}
            \item Remove \(v_i^{(l)}\) and all incident edges from \(G_l\):  
            \( N_l \leftarrow N_l \setminus \{v_i^{(l)}\} \).
            \item Select an organization \(G_r\), \(r \ne l\), uniformly at random.
            \item Add the inventor to \(G_r\) with no edges:  
            \( N_r \leftarrow N_r \cup \{v_i^{(r)}\} \).  
            The inventor retains the same ID \(i\) and knowledge set \(B_i\).
            \item If an inventor returns to a previously visited organization, they rejoin with no pre-existing edges.
        \end{itemize}
        \item The movement event (departure and destination) is recorded for output and plotting.
    \end{enumerate}

    \item \textbf{Facilitator Temporary Edge Addition.}
    Facilitator \texttt{tg\(_l\)} may add up to \(k_l\) temporary edges to \(G_l\), where \(k_l \ge 0\) is a parameter of the simulation. All temporary edges are removed at the end of the stage.

    \item \textbf{Clique Selection and Knowledge Creation.}
    \begin{enumerate}
        \item A clique \(C \subseteq N_l\) is selected uniformly at random from all cliques in \(G_l\), considering both permanent and temporary edges.
        \item A new knowledge item is created for clique \(C\) with probability:
        \[
        P_{\text{new}} = 
        \begin{cases}
        \dfrac{\left|\bigcup_{v_i^{(l)} \in C} B_i\right|}{|K_l|}, & \text{if } |K_l| > 0, \\
        1, & \text{if } |K_l| = 0 .
        \end{cases}
        \]
        \item If a new item \(x\) is created:
        \begin{itemize}
            \item Add \(x\) to each \(B_i\) for all \(v_i^{(l)} \in C\).
            \item Record \(\text{origin}(x) = l\).
        \end{itemize}
    \end{enumerate}

    \item \textbf{Facilitator Reward.}
    Facilitator \texttt{tg\(_l\)} receives a reward of 1 for the stage if the selected clique contains at least one temporary edge added by \texttt{tg\(_l\)}.

    \item \textbf{Shapley Value Computation.}
    For organization \(G_l\), define the cooperative game
    \[
    (N_l, v), \qquad v(S) = \left|\bigcup_{v_i^{(l)} \in S} B_i\right| \quad \forall\, S \subseteq N_l.
    \]
    The Shapley value of an inventor \(v_i^{(l)}\) is computed as:
    \[
    \phi_i^{(l)} = \sum_{x \in B_i} \frac{1}{|\{v_j^{(l)} \in N_l : x \in B_j\}|}.
    \]

    \item \textbf{Global Statistics Update.}
    After every stage:
    \begin{itemize}
        \item Global cumulative statistics for all inventors are updated.
        \item Plots for inventors that changed organizations show continuous curves, with color changes marking organization transitions.
    \end{itemize}
\end{enumerate}

\subsubsection*{Invariant}

Implement the following invariant in the code so that after each simulation stage the invariant expressed in the lemma below will be checked and a log will be produced to a log file that details the sequence number of the simulation, the organization that participated, the coalition that was chosen, if it had a tg edge added, and wether or not the invariant below held.  

\begin{lemma}[Correctness Invariant]
\label{lemma:correctness}
Let $B_v$ denote the knowledge set of inventor $v$.
After every stage of the simulation, we have
\[
\left| \bigcup_v B_v \right|
\;=\;
noTGKnowldge \;+\; \sum_{\text{all facilitators } tg} r(tg),
\]
where:
\begin{itemize}
    \item $noTGKnowldge$ counts the number of knowledge items created in stages where the selected clique contained \emph{no} temporary edge added by any facilitator.
    \item For each facilitator $tg$, the value $r(tg)$ counts the number of knowledge items created in stages where the selected clique contained at least one temporary edge added by that $tg$.
\end{itemize}
Hence, each created knowledge item is counted exactly once by exactly one of the two terms.
\end{lemma}


\textbf{Simulation Task One:}
\begin{itemize}
    \item In this simulation the number of graphs is \( 1 \) 
    \item Initialize a graph with five inventors:
    \begin{itemize}
        \item \( v_1 \) connected to \( v_2 \) and \( v_3 \)
        \item \( v_2 \) connected to \( v_3 \)
        \item \( v_3 \) connected to \( v_4 \) and \( v_5 \)
        \item \( v_4 \) connected to \( v_5 \)
    \end{itemize}
    \item Include a facilitator player \texttt{tg} with \( k = 2 \).
    \item Run the simulation for 10 stages.
    \item In one simulation, let \texttt{tg} choose two edge randomly.
    \item In another simulation, let \texttt{tg} always add the edge \( (v_1, v_4), (v_1, v_5) \).
    \item Compute and plot:
    \begin{itemize}
        \item The cumulative reward of \texttt{tg} after each stage.
        \item The cumulative growth of knowledge in the organization after each stage.
        \item A single bar that represents the entire knowledge in the organization and the percentage of the knowledge that was created without the help of the \texttt{tg}, represented by the  $noTGKnowldge$ variable, and the percentage of the knowledge that was created with the help of the \texttt{tg}.  To that end the formula of the invariant lemma \ref{lemma:correctness} is used.  

    \end{itemize}
    \item Determine which $tg$ strategy created more knowledge and which $tg$ strategy gave more reward to the $tg$. 

\end{itemize}

Each plot should be clearly identified clearly associated with the simulation task and strategy used by the $tg_i$s.  In addition the plots should be both written to a file and standard output.  

\textbf{Simulation Task Two:}
\begin{itemize}
    \item In this simulation the number of graphs is \( 1 \) 
    \item Initialize a graph with ten inventors:
    \begin{itemize}
        \item \( v_1 \) to \( v_5 \) connected to each other.
        \item \( v_5 \) to \( v_{10} \) connected to each other.
    \end{itemize}
    \item Include a facilitator player \texttt{tg} with \( k = 5 \).
    \item Run the simulation for 10 stages.
    \item In one simulation, let \texttt{tg} choose to connect 5 random pairs of currently unconnected vertices (e.g., \( v_1 \) with \( v_{10} \)).
    \item In another simulation, let \texttt{tg} always connect \( v_4 \) with \( v_6, v_7, v_8, v_9, v_{10} \).
    \item Compute and plot:
    \begin{itemize}
        \item The cumulative reward of \texttt{tg} after each stage.
        \item The cumulative knowledge growth in the organzation after each stage.
        \item A single bar that represents the entire knowledge in the organization and the percentage of the knowledge that was created without the help of the \texttt{tg}, represented by the  $noTGKnowldge$ variable, and the percentage of the knowledge that was created with the help of the \texttt{tg}.  To that end the formula of the invariant lemma \ref{lemma:correctness} is used.  

    \end{itemize}
    \item Determine which $tg$ strategy created more knowledge and which $tg$ strategy gave more reward to the $tg$. 
\end{itemize}

\textbf{Prompt emphasis}
I'm missing the conclusions in each simulation task - which \texttt{tg} strategy created more knowledge? Which \texttt{tg} strategy created more reward to the \texttt{tg}? Please render the code again with the conclusions at the end of the simulation.

\textbf{Prompt emphasis}
I'm missing which strategy gave the \texttt{tg} more reward.  Please render the code again with the conclusions at the end of the simulation.

\textbf{Prompt - following on AI suggestion}
Please Add Monte-Carlo averaging so conclusions are statistically stable

\


\end{humanprompt}

\subsection{Generated simulation}

Obtained correct simulation in two stages, \href{https://colab.research.google.com/drive/1JQ_SYM3DdBjVKXXYQqewHwbaWsxR3lz4#scrollTo=Xs2Qcd24q-nn&line=1&uniqifier=1}{one with} and \href{https://colab.research.google.com/drive/1JQ_SYM3DdBjVKXXYQqewHwbaWsxR3lz4#scrollTo=7uU-RXw-phci&line=3&uniqifier=1}{one without} Monte Carlo simulation that ensures stability of the result.  Had to emphasis that I want a final result concluded by the simulation several times - see prompt above.  Also followed the suggestion from AI to implement
the Monte Carlo simulation to ensure statistical stability.



\section{Ambiguous business specification}
\label{sec:ambiguous_business_specification}

\begin{itemize}
    \item We assume that the business analyst wants to completely control the business logic.
    \item To simulate the undesired situation in which the business knowledge is ambiguous and determine the effect, we use a summary of the original specification.
    \item We observe that different aspects of the simulation are not as desired, in particular the notion of organizations and graph structure, the association of tg to an organization, and the order of events.
\end{itemize}

Our setup is the RooCode framework and extension for the VSCode IDE, which enables us to interact with LLMs (in a BYOM manner) with seamless integration with the codebase. The model of our choice for this experiment is \emph{Claude-Sonnet-4.5}.

Firstly, I begin with the original, refined specification from Section~\ref{comparing_strategies} (referred to below as the file \verb|humanprompt.tex|). In this specification, we fully specify the business logic that we want to control. This is the input that we send to the model, which we care about.   It simulates a situation in a commercial application in which the business analyst determines the way some business transaction is done, e.g., an insurance claim, and wants to completely control the manner in which the transaction is carried out.  The rest of the program, mainly the user interface, is less critical, and we therefore give the model greater freedom. Later, we will use a summary and obtain an ambiguous version of the business logic specification, to simulate what happens if we do not completely specify the business logic that we wish to control.

\begin{humanprompt}
\paragraph{Prompt.}
My task for you is written in latex format in the file \verb|humanprompt.tex|.
You may use intermediate (Markdown) files for problem specification, design overview, task division, and you may ask me for my intervention at any time.
Implement the code using Python. For the UI, you may use simple HTML and CSS, and JS when necessary, but never node.js.
Be clear in your intentions and record all decisions you take in (one or more) files.
Make sure the result is actionable and verifiable by running and looking at the UI or logs.

\end{humanprompt}

Notice the sentence "You may use intermediate (Markdown) files for problem specification, design overview, task division." In other words, I permit the model to use notions coming from Spec-Driven Development (SDD) to guide its own behavior, by first formalizing the problem in intermediate files. Not only was it shown that it improves code generation, it also gives transparency into the decisions that the model has taken with regard to peripheral elements that surround the business logic.

The model successfully created a web-based application with basic functionality, mainly demonstrating the results of the three simulations, while maintaining the invariant during all steps.

\begin{aireply}

\paragraph{Response.}
Implemented the full simulation project per the LaTeX specification.

\paragraph{Created files}

\begin{itemize}
    \item \verb|design.md| – detailed design documentation.
    \item \verb|requirements.txt| – Python dependencies (\verb|networkx|, \verb|matplotlib|).
    \item \verb|simulation.py| – complete Python simulation with graph handling, facilitator logic, invariant checking, logging, and plotting for all four scenarios.
    \item \verb|simulation.log| – log file with per‑stage information and invariant verification.
    \item \verb|reward_plot.png|, \verb|knowledge_plot.png|, \verb|knowledge_split.png| – generated plots for cumulative facilitator reward, cumulative knowledge size, and final knowledge split.
    \item \verb|index.html| – static HTML UI displaying the three plots and loading the log.
    \item \verb|styles.css| – basic styling for the UI.
\end{itemize}

\end{aireply}

The model has created a \verb|design.md| file to guide itself throughout code generation.

Next, we introduce ambiguity by two steps. The first step is transforming the carefully written math-rich specification into a pure natural language description, which is described below. The next step will be obtaining a summary from the natural language description.

We begin with asking the model's help with winding down the formalism in the specification.

\begin{humanprompt}
\paragraph{Prompt.}
Reconsidering my original intent written in LaTeX \verb|humanprompt.tex|, how would you rephrase it in natural language in a way that doesn't contain math symbols at all? something that one would say out loud while presenting "without slides". Write that in a new file called "humanprompt.nl.md"
\end{humanprompt}

\begin{aireply}
\paragraph{Response.}
Created \verb|humanprompt.nl.md| - a natural language version of the LaTeX specification that can be read aloud without mathematical notation.

\paragraph{Key transformations:}

\begin{itemize}
    \item Mathematical symbols replaced with plain English (e.g., $m \geq 1$ → "one or more organizations")
    \item Set notation converted to everyday language (e.g., "union of knowledge sets" → "combined knowledge of everyone")
    \item Formulas explained conceptually (e.g., probability formulas described as ratios in words)
    \item Graph theory terms simplified (e.g., "clique" explained as "complete subgroups where everyone is connected")
    \item Lemma and formal notation converted to plain statements
\end{itemize}

The document maintains all the technical details and requirements from the original LaTeX specification while being accessible for verbal presentation without slides or visual aids. It flows as a narrative explanation suitable for someone listening to understand the simulation model.

\end{aireply}

Indeed, the order of events in each stage turns out to be correct as follows.
\begin{enumerate}
    \item Pick an Organization
    \item Someone Might Leave
    \item Facilitator Adds Temporary Connections
    \item Creating New Knowledge
    \item Reward the Facilitator
    \item Calculate Contributions
    \item Update Statistics
\end{enumerate}

However, it is still too long for someone to write it on their first try. Thus I asked that it be made briefer:

\begin{humanprompt}

\paragraph{Prompt.}
Thanks. Now create a less structured description, and a briefer one, even if it loses some nuances. Make it at most three paragraphs. Save to "humanprompt.nl.txt"
\end{humanprompt}

\begin{aireply}

\paragraph{Response.}
Created \verb|humanprompt.nl.txt| - a brief, three-paragraph natural language summary of the simulation specification.

\paragraph{Content overview:}

\begin{itemize}
    \item Paragraph 1: Core mechanics - inventors in networks, knowledge creation through collaboration, mobility between organizations
    \item Paragraph 2: Facilitator role - temporary connections, rewards, and the invariant check
    \item Paragraph 3: Experimental design - two scenarios, two strategies, comparative analysis The description sacrifices detailed specifications for readability and brevity, focusing on the essential concepts someone would need to understand what the simulation does at a high level.
\end{itemize}
\end{aireply}

Here is the content of \verb|humanprompt.nl.txt|:

\begin{aireply}

\paragraph{Content of humanprompt.nl.txt:}

We're simulating how knowledge gets created when inventors collaborate in a network. Each inventor has their own collection of knowledge, and they're connected to other inventors through a social network graph. At each time step, we pick a random group of fully-connected inventors (a clique) and they might create new knowledge together - the probability depends on how much combined knowledge they already have versus the total in the organization. Sometimes inventors with lots of knowledge might leave and join a different organization, losing their network connections but keeping their knowledge.

There's also a facilitator player who can temporarily add new connections between inventors for one time step, trying to help create more knowledge. The facilitator gets rewarded when knowledge is created by a group that includes one of their temporary connections. We track how much knowledge was created with the facilitator's help versus without, and this should always add up to the total knowledge in the system - that's our invariant check.

We run two scenarios with different network structures and compare two facilitator strategies: one that randomly connects inventors, and one that always makes the same strategic connections. For each run we plot how the facilitator's reward grows over time, how total knowledge grows, and show the final split between facilitated versus organic knowledge creation. The goal is to see which strategy helps create more knowledge overall and which earns the facilitator more reward points.

\end{aireply}

Now I wanted to attempt the original intent once more, with this condensed prompt instead of the original math-rich one.

The remainder of the conversation happens in a new context window and a directory cleared from the previous attempt.
Here, too, I go with an SDD-inspired approach:

\begin{humanprompt}
You must NOT read any file ending with \verb|.tex| or \verb|.md|. Your task is to read \verb|humanprompt.nl.txt| and to generate a more rigorous and full description, one that can use math symbols and formulas to be consistent and unambiguous. Write it in latex, in a new file "convertedprompt.tex"
\end{humanprompt}

\begin{aireply}
\paragraph{Response.}
I've created \verb|convertedprompt.tex| with a rigorous mathematical specification of the knowledge creation simulation. The document includes:

\begin{enumerate}
    \item \textbf{System Model} - Formal definitions of the inventor network graph, knowledge states, clique selection, and knowledge creation probability
    \item \textbf{Inventor Mobility} - Mathematical description of the departure process
    \item \textbf{Facilitation Mechanism} - Precise definitions of facilitator actions, knowledge attribution, and reward function
    \item \textbf{System Invariant} - The accounting equation ensuring facilitated + organic knowledge equals total knowledge
    \item \textbf{Experimental Design} - Formal specification of network topologies, facilitator strategies (random vs. fixed), performance metrics, and objectives
\end{enumerate}

All concepts are expressed using mathematical notation with proper set theory, probability functions, and graph theory formalism to eliminate ambiguity.
\end{aireply}

Unfortunately, despite the model's optimism expressed in the final sentence, looking at the resulting specification several problems popped up:

\begin{enumerate}
    \item The notion of organizations almost disappeared.
    \item There were mistakes in the graph structure, such as a confusion between directed and undirected edges and the definition of a clique, which introduce ambiguity.
    \item The association of tg to an organization was not touched upon.
    \item The order of events was problematic, as the choosing of a clique happens \emph{before} the introduction of new, temporary edges by tg.
    \item The budget of tg was missing, it seems that tg has an infinite budget.
    \item The specificity of the scenarios was lost after the summary.
\end{enumerate}

Finally, I ask it to generate a similar web-based simulator for the specified business logic.

\begin{humanprompt}
\paragraph{Prompt.}
Write python code to implement simulations for the task described in \verb|convertedprompt.tex|. I want a basic html-css-js (vanilla) UI that shows the simulation result (how much knowledge there is at the end, split by knowledge that was faciliated by tg and knowledge that wasn't. also, at each step, it must be checked that the invariant holds, and a line to a log file must be added. then the log file must also be presented in the UI).
\end{humanprompt}

See a screenshot of the result in Figure~\ref{fig:knowledge_screenshot}.

\begin{figure}[h]
    \centering
    \includegraphics[width=0.75\linewidth]{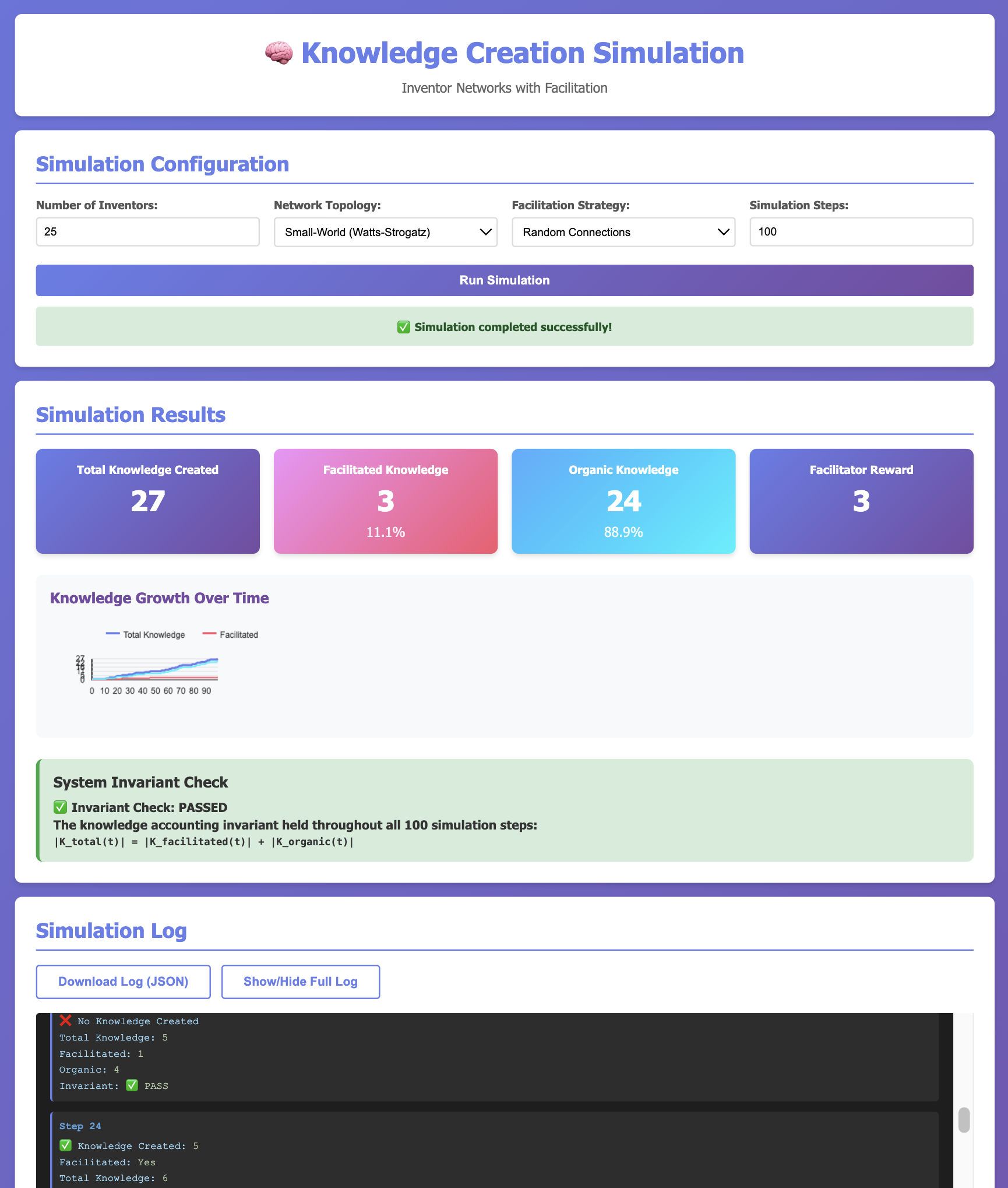}
    \caption{Knowledge Creation Simulation: screenshot of the result.}
    \label{fig:knowledge_screenshot}
\end{figure}

\paragraph{What we learned.}
The benefit from this experiment is twofold.

\paragraph{1. Controlled vs. not controlled aspects.}
A system analyst will want to separate between what the model can modify and what it must preserve. In other words, the art in the process is to know where and how to draw the line between controlled and not controlled program aspects. The specification of a GUI and other peripheral elements is not core logic and can be modified or expanded by the model without prompting the system analyst. However, anything pertaining to the business logic should be controlled by the analyst; it must be returned to the analyst for refinement or confirmation, especially in cases where the intention is ambiguous. This idea of controlled vs. not controlled aspects is captured in Figure~\ref{fig:control}.

\begin{figure}[h]
    \centering
    \includegraphics[width=0.6\linewidth]{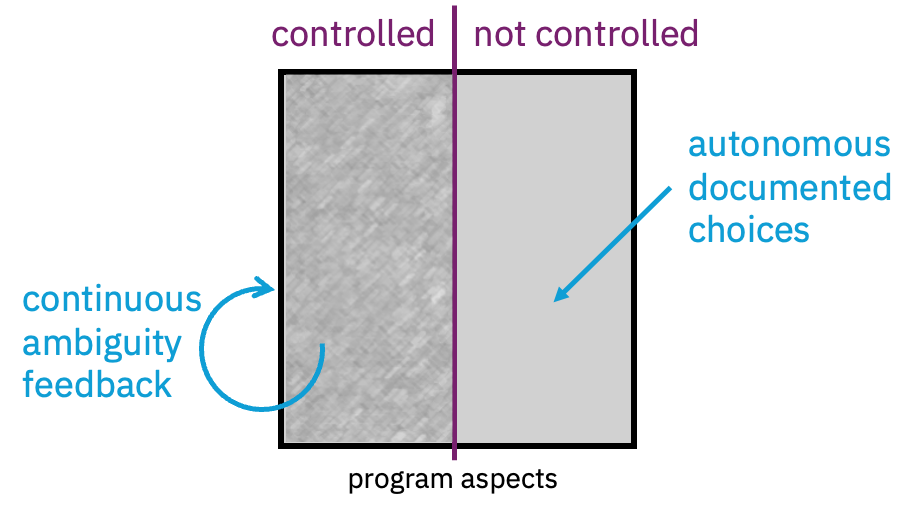}
    \caption{Controlled vs. not controlled program aspects.}
    \label{fig:control}
\end{figure}

\paragraph{2. Primacy of precision.}
Typically, the first iteration of a human-written description of a business logic is not precise. The line between controlled and not controlled variability of a program is blurry and not strictly stated. The summarized description generated above is a simulation of this: bad choice of words, confusing prepositions, and missing key information all contribute to that blurry line, and 
give the model space to be creative
in more aspects of the program than was originally intended by the system analyst.

\section{Related Work}

Recent advances in software engineering have explored the use of
Large Language Models (LLMs) and agentic systems to support program
synthesis and multi-step reasoning tasks. Agentic systems leverage
reasoning-action loops and multi-agent coordination to generate,
verify, and correct code autonomously, providing capabilities similar
to those described in Section~1, while remaining dependent on human
supervision for reliability. VIB (Value-Intention-Behavior) programming
offers a complementary perspective, structuring agent behavior around
high-level goals, intentions, and concrete actions \cite{ge2025surveyvibecodinglarge}, which
aligns well with the layered reasoning used in AI-assisted software generation.

Spec-Driven Development (SDD) provides a methodology for guiding both
human and AI developers from a structured specification to verified
code \cite{SDD2025}. By defining \emph{what} and \emph{why} first,
SDD improves alignment, traceability, and maintainability of
AI-generated outputs. Our work builds on this concept by using natural
language augmented with mathematical notation as an intermediate
representation. This approach retains the clarity and precision of
formal specifications while significantly reducing the notation
overhead compared to full formal languages such as \textsf{Z}.

The \textsf{Z} specification language remains a foundational formal
method for mathematically precise modeling of system state, operations,
and invariants \cite{Z1990}. Unlike \textsf{Z}, our intermediate
representation trades some formal rigor for readability and ease of use,
enabling AI agents to interpret and act on specifications directly
without requiring deep mathematical expertise. In this way, our
approach can be seen as a simplification of formal specification that
balances precision, expressiveness, and AI accessibility, bridging the
gap between rigorous formal methods and practical AI-assisted software
development.

\section{Conclusion}

This work illustrates that AI-assisted software development can
benefit from structured specifications to guide code generation and
support maintainable designs, though several limitations remain. By
combining natural language with lightweight mathematical notation, our
approach captures some of the advantages of formal specification—such
as explicit invariants and reasoning about correctness—while reducing
the complexity and overhead of full formal methods like \textsf{Z}.
The intermediate representation helped guide AI agents in generating
code, tests, and supporting infrastructure, drawing on principles from
Spec-Driven Development (SDD) and VIB programming to improve
alignment and traceability. However, the simulation used as a case
study, while logically complex, represents only a small piece of
software, and challenges such as attention limits, specification
decomposition, and iterative refinement were encountered during
development.

Future research directions include extending these methods to fully
exploit natural language specifications without embedded mathematics,
exploring richer abstractions, modular interfaces, and
document-driven composition to simplify system specification. Such
investigations could further reduce the barrier to entry for
formal-specification-inspired practices, broaden the applicability of
AI-assisted software engineering, and provide more scalable solutions
for larger and more complex software systems.

\section{Acknowledgments}

The author would like to thank Navid Asgari (Fordham University,
Gabelli School of Business) and Deepak Nayak (Ohio State University,
Fisher College of Business) for introducing and discussing foundational
ideas related to knowledge in organizations. Their insights and
conversations were instrumental in motivating the problem domain and
shaping the conceptual framing of the case study presented in this
paper, with \emph{tertius iungens} (the third that connects) playing a
key role in the simulation.

\bibliographystyle{plain}  
\bibliography{main}  

\appendix
\section*{Appendix}
\addcontentsline{toc}{section}{Appendix}

\renewcommand{\thesection}{\Alph{section}}

\section{Prompt for generating a simulation with an invariant}
\label{promptGeneratingInvariant}

In what follows I detail the prompt for a clean new generation of the simulation with its associated invariants. Following is hte prompt. 

\begin{humanprompt}

\textbf{Please generate a python program that implements the following simulation specification and its associated invariant.}\\

\subsection{More than one organization - adding back simulation tasks}
\label{twoOrgSpec}

This subsection defines a simulation model for knowledge creation and diffusion across multiple organizations. Each organization is represented as a graph, and inventors may move between organizations over time. The goal is to study how network structure, mobility, and facilitator interventions influence knowledge growth and individual contribution.

\subsubsection*{Model Entities and Notation}

\paragraph{Organizations and Inventors.}
There are \(m \ge 1\) organizations, each represented as a graph \(G_l = (N_l, E_l)\) for \(l = 1, \ldots, m\). An inventor is an individual with a \textbf{global unique identifier}. If inventor \(i\) is currently in organization \(l\), we denote this inventor by \(v_i^{(l)}\). The subscript \(i\) never changes, even if the inventor moves between organizations; the superscript indicates the current organization membership.

\paragraph{Knowledge Sets.}
Each inventor \(v_i^{(l)}\) owns a set of knowledge items \(B_i\) (which persists with the inventor globally, regardless of organization). The set of knowledge in organization \(G_l\) at a given time is defined as
\[
K_l = \bigcup_{v_i^{(l)} \in N_l} B_i .
\]

\paragraph{Knowledge Attribution.}
Each knowledge item is attributed to the organization in which it was first created. If knowledge item \(x\) was generated in organization \(G_l\), its origin is recorded as \(\text{origin}(x) = l\). This information will later be used to compute, for each inventor, the percentage of their knowledge originating from each organization.

\paragraph{Facilitators.}
Each organization \(G_l\) has an associated facilitator denoted \texttt{tg\(_l\)}. Facilitator \texttt{tg\(_l\)} may add temporary edges to \(G_l\) only; it cannot modify other organizations.

\paragraph{Global Stage.}
Time evolves in discrete \textbf{global stages}. In each stage, exactly one organization is selected uniformly at random to execute the stage operations. Some operations (e.g., inventor movement) may affect multiple organizations.

\subsubsection*{Stage Operations}

Each global stage consists of the following sequence:

\begin{enumerate}
    \item \textbf{Organization Selection.}  
    Select an organization \(G_l\) uniformly at random from \(\{G_1, \ldots, G_m\}\). All subsequent steps in the stage apply with respect to \(G_l\).

    \item \textbf{Leaving-Organization Step (with probability \(p = 0.2\)).}
    This step is executed only if \(m > 1\).
    \begin{enumerate}
        \item Identify inventor(s) in \(G_l\) with maximal knowledge size:
        \[
        M_l = \left\{ v_i^{(l)} \in N_l \mid |B_i| = \max_{v_j^{(l)} \in N_l} |B_j| \right\}.
        \]
        \item If \(|M_l| > 1\), select one inventor uniformly at random from \(M_l\); denote this inventor by \(v_i^{(l)}\).
        \item Compute the probability of leaving:
        \[
        P_{\text{leave}} = 
        \begin{cases}
        \dfrac{|B_i|}{|K_l|}, & \text{if } |K_l| > 0, \\
        0, & \text{if } |K_l| = 0 .
        \end{cases}
        \]
        \item With probability \(P_{\text{leave}}\), the inventor leaves \(G_l\):
        \begin{itemize}
            \item Remove \(v_i^{(l)}\) and all incident edges from \(G_l\):  
            \( N_l \leftarrow N_l \setminus \{v_i^{(l)}\} \).
            \item Select an organization \(G_r\), \(r \ne l\), uniformly at random.
            \item Add the inventor to \(G_r\) with no edges:  
            \( N_r \leftarrow N_r \cup \{v_i^{(r)}\} \).  
            The inventor retains the same ID \(i\) and knowledge set \(B_i\).
            \item If an inventor returns to a previously visited organization, they rejoin with no pre-existing edges.
        \end{itemize}
        \item The movement event (departure and destination) is recorded for output and plotting.
    \end{enumerate}

    \item \textbf{Facilitator Temporary Edge Addition.}
    Facilitator \texttt{tg\(_l\)} may add up to \(k_l\) temporary edges to \(G_l\), where \(k_l \ge 0\) is a parameter of the simulation. All temporary edges are removed at the end of the stage.

    \item \textbf{Clique Selection and Knowledge Creation.}
    \begin{enumerate}
        \item A clique \(C \subseteq N_l\) is selected uniformly at random from all cliques in \(G_l\), considering both permanent and temporary edges.
        \item A new knowledge item is created for clique \(C\) with probability:
        \[
        P_{\text{new}} = 
        \begin{cases}
        \dfrac{\left|\bigcup_{v_i^{(l)} \in C} B_i\right|}{|K_l|}, & \text{if } |K_l| > 0, \\
        1, & \text{if } |K_l| = 0 .
        \end{cases}
        \]
        \item If a new item \(x\) is created:
        \begin{itemize}
            \item Add \(x\) to each \(B_i\) for all \(v_i^{(l)} \in C\).
            \item Record \(\text{origin}(x) = l\).
        \end{itemize}
    \end{enumerate}

    \item \textbf{Facilitator Reward.}
    Facilitator \texttt{tg\(_l\)} receives a reward of 1 for the stage if the selected clique contains at least one temporary edge added by \texttt{tg\(_l\)}.

    \item \textbf{Shapley Value Computation.}
    For organization \(G_l\), define the cooperative game
    \[
    (N_l, v), \qquad v(S) = \left|\bigcup_{v_i^{(l)} \in S} B_i\right| \quad \forall\, S \subseteq N_l.
    \]
    The Shapley value of an inventor \(v_i^{(l)}\) is computed as:
    \[
    \phi_i^{(l)} = \sum_{x \in B_i} \frac{1}{|\{v_j^{(l)} \in N_l : x \in B_j\}|}.
    \]

    \item \textbf{Global Statistics Update.}
    After every stage:
    \begin{itemize}
        \item Global cumulative statistics for all inventors are updated.
        \item Plots for inventors that changed organizations show continuous curves, with color changes marking organization transitions.
    \end{itemize}
\end{enumerate}

\subsubsection*{Invariant}

Implement the following invariant in the code so that after each simulation stage the invariant expressed in the lemma below will be checked and a log will be produced to a log file that details the sequence number of the simulation, the organization that participated, the coalition that was chosen, if it had a tg edge added, and wether or not the invariant below held.  

\begin{lemma}[Correctness Invariant]
\label{lemma:correctness}
Let $B_v$ denote the knowledge set of inventor $v$.
After every stage of the simulation, we have
\[
\left| \bigcup_v B_v \right|
\;=\;
noTGKnowldge \;+\; \sum_{\text{all facilitators } tg} r(tg),
\]
where:
\begin{itemize}
    \item $noTGKnowldge$ counts the number of knowledge items created in stages where the selected clique contained \emph{no} temporary edge added by any facilitator.
    \item For each facilitator $tg$, the value $r(tg)$ counts the number of knowledge items created in stages where the selected clique contained at least one temporary edge added by that $tg$.
\end{itemize}
Hence, each created knowledge item is counted exactly once by exactly one of the two terms.
\end{lemma}


\textbf{Simulation Task One:}
\begin{itemize}
    \item In this simulation the number of graphs is \( 1 \) 
    \item Initialize a graph with five inventors:
    \begin{itemize}
        \item \( v_1 \) connected to \( v_2 \) and \( v_3 \)
        \item \( v_2 \) connected to \( v_3 \)
        \item \( v_3 \) connected to \( v_4 \) and \( v_5 \)
        \item \( v_4 \) connected to \( v_5 \)
    \end{itemize}
    \item Include a facilitator player \texttt{tg} with \( k = 1 \).
    \item Run the simulation for 10 stages.
    \item In one simulation, let \texttt{tg} choose an edge randomly.
    \item In another simulation, let \texttt{tg} always add the edge \( (v_1, v_4) \).
    \item Compute and plot:
    \begin{itemize}
        \item The cumulative reward of \texttt{tg} after each stage.
        \item The knowledge growth of each inventor after each stage.
        \item The Shapley value \( \phi_i \) of each inventor after each stage.

    \end{itemize}
\end{itemize}

\textbf{Simulation Task Two:}
\begin{itemize}
    \item In this simulation the number of graphs is \( 1 \) 
    \item Initialize a graph with ten inventors:
    \begin{itemize}
        \item \( v_1 \) to \( v_5 \) connected to each other.
        \item \( v_5 \) to \( v_{10} \) connected to each other.
    \end{itemize}
    \item Include a facilitator player \texttt{tg} with \( k = 5 \).
    \item Run the simulation for 10 stages.
    \item In one simulation, let \texttt{tg} choose to connect 5 random pairs of currently unconnected vertices (e.g., \( v_1 \) with \( v_{10} \)).
    \item In another simulation, let \texttt{tg} always connect \( v_4 \) with \( v_6, v_7, v_8, v_9, v_{10} \).
    \item Compute and plot:
    \begin{itemize}
        \item The cumulative reward of \texttt{tg} after each stage.
        \item The knowledge growth of each inventor after each stage.
        \item The Shapley value \( \phi_i \) of each inventor after each stage.

    \end{itemize}
\end{itemize}

\textbf{Simulation Task Three:}

\begin{itemize}
    \item In this simulation the number of graphs is \( 2 \) 
    \begin{itemize}
        \item Initialize \( G_1 \) with ten inventors:
        \begin{itemize}
            \item \( v^1_1 \) to \( v^1_5 \) connected to each other.
            \item \( v^1_5 \) to \( v^1_{10} \) connected to each other.
        \end{itemize}
    \end{itemize}
    \begin{itemize}
        \item Initialize \( G_2 \) with ten inventors:
        \begin{itemize}
            \item \( v^2_1 \) to \( v^2_5 \) connected to each other.
            \item \( v^2_5 \) to \( v^2_{10} \) connected to each other.
        \end{itemize}
    \end{itemize}            
    \item Include a facilitator player \texttt{\(tg_1\)} with \( k = 5 \).  Let \texttt{\(tg_1\)} choose to connect 5 random pairs of currently unconnected vertices (e.g., \( v^1_1 \) with \( v^1_{10} \)).
    \item Include a facilitator, \texttt{\(tg_2\)} with \( k = 1 \).  Let \texttt{\(tg_2\)} choose to connect 1 random pairs of currently unconnected vertices (e.g., \( v^1_1 \) with \( v^1_{10} \)).
    \item Compute and plot:
    \begin{itemize}
        \item The cumulative reward of \texttt{\(tg_1\)} after each stage.
        \item The cumulative reward of \texttt{\(tg_2\)} after each stage.
        \item The knowledge growth of each inventor after each stage.

    \end{itemize}
    \item Run the simulation for 10 stages.

\end{itemize}

Each plot should be clearly identified clearly associated with the simulation task and strategy used by the $tg_i$s.  In addition the plots should be both written to a file and standard output.

\
When there is more than one graph, create a table that represents what percentage of an inventor knowledge was created in each of the organizations.

\end{humanprompt}

\end{document}